\documentclass[final]{IEEEtran}

\usepackage{color}
\usepackage{csquotes}
\usepackage{enumitem}
\usepackage{float}

\floatstyle{ruled}
\newfloat{algorithm}{tbp}{loa}
\providecommand{\algorithmname}{Algorithm}
\floatname{algorithm}{\protect\algorithmname}
\usepackage{hyperref}
\usepackage{amsopn}

\DeclareMathOperator*{\argmin}{arg\,min}

\makeatletter
\newcommand{\manuallabel}[2]{\def\@currentlabel{#2}\label{#1}}
\makeatother
\usepackage{amsmath}
\usepackage{amssymb}
\usepackage{amsthm}
\usepackage{bbm}
\usepackage{tikz}
\usepackage{mathrsfs}
\usepackage{pgfplots}
\usepackage{cite}
\pgfplotsset{compat=1.14}
\usetikzlibrary{arrows}
\usetikzlibrary{arrows.meta}
\newtheorem{lemma}{Lemma}
\newtheorem{theorem}{Theorem}
\newtheorem{dfn}{Definition}
\newtheorem{remark}{Remark}
\newtheorem{cor}{Corollary}
\newtheorem{conj}{Open Problem}
\newtheorem{assump}{Assumption}

\title{Optimal Rates of Teaching and \\ Learning Under Uncertainty}
\author{Yan Hao Ling and Jonathan Scarlett}

\newcommand{\p}{\mathbb{P}}
\newcommand{\e}{\mathbb{E}}
\newcommand{\bh}{Bhattacharyya }
\newcommand{\mutilde}{\tilde{\mu}}

\begin{document}
    \maketitle
    \begin{abstract}
        In this paper, we consider a recently-proposed model of teaching and learning under uncertainty, in which a teacher receives independent observations of a single bit corrupted by binary symmetric noise, and sequentially transmits to a student through another binary symmetric channel based on the bits observed so far.  After a given number $n$ of transmissions, the student outputs an estimate of the unknown bit, and we are interested in the exponential decay rate of the error probability as $n$ increases.  We propose a novel block-structured teaching strategy in which the teacher encodes the number of 1s received in each block, and show that the resulting error exponent is the binary relative entropy $D\big(\frac{1}{2}\|\max(p,q)\big)$, where $p$ and $q$ are the noise parameters.  This matches a trivial converse result based on the data processing inequality, and settles two conjectures of [Jog and Loh, 2021] and [Huleihel, Polyanskiy, and Shayevitz, 2019].  In addition, we show that the computation time required by the teacher and student is linear in $n$. We also study a more general setting in which the binary symmetric channels are replaced by general binary-input discrete memoryless channels. We provide an achievability bound and a converse bound, and show that the two coincide in certain cases, including (i) when the two channels are identical, and (ii) when the student-teacher channel is a binary symmetric channel.  More generally, we give sufficient conditions under which our learning rate is the best possible for block-structured protocols.
    \end{abstract}
    
    \long\def\symbolfootnote[#1]#2{\begingroup\def\thefootnote{\fnsymbol{footnote}}\footnote[#1]{#2}\endgroup}
    
    \symbolfootnote[0]{The authors are with the  Department of Computer Science and Department of Mathematics, School of Computing, National University of Singapore (NUS). Jonathan Scarlett is also with the Institute of Data  Science, NUS. Emails: \url{e0174827@u.nus.edu};  \url{scarlett@comp.nus.edu.sg}.
    
    This work was presented in part at the 2021 IEEE International Symposium on Information Theory (ISIT).}
    
    \section{Introduction}
    
    In several societal and technological domains, one is interested in how agents interact with their environment and with each other to attain goals such as learning information about the environment, conveying this information to other agents, and reaching a common consensus. While a comprehensive theoretical understanding of such problems would likely require highly sophisticated mathematical models, even the simplest models come with unique insights and challenges.
    
    In this paper, we study a recently-proposed model of teaching and learning under uncertainty \cite{jog2020teaching,onebit}, in which a teacher observes noisy information regarding an unknown 1-bit quantity $\Theta$ (the environment), and seeks to convey information to a student via a noisy channel to facilitate learning $\Theta$.  We establish the optimal learning rate (i.e., exponential decay of the error probability) of this problem, thereby settling two notable conjectures made in \cite{jog2020teaching,onebit} described in detail below.  In addition, we study a generalization of  the problem to more general binary-input discrete memoryless channels, and obtain the optimal learning rate in certain special cases of interest.
    
    \subsection{Model and Definitions}
    \label{model-dfn}
    
    \begin{figure}
        \begin{centering}
            \includegraphics[width=0.99\columnwidth]{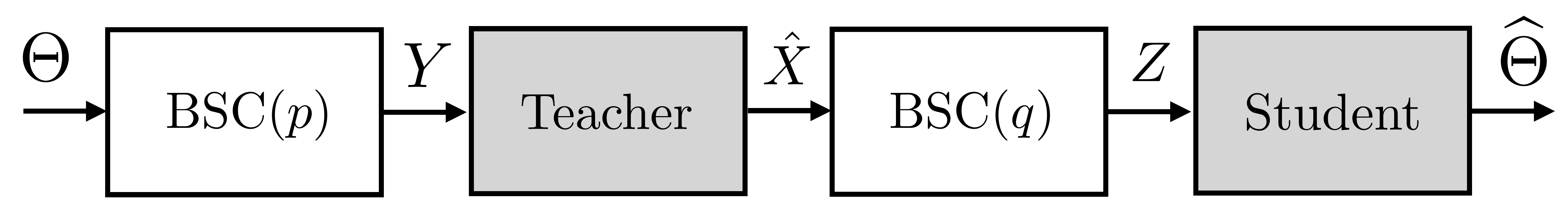}
            \par
        \end{centering}
        
        \caption{Illustration of the problem setup. \label{fig:setup}}
    \end{figure}
    
    We first formalize the model, which is depicted in Figure \ref{fig:setup}.  The unknown quantity of interest is a binary random variable $\Theta$ that takes either value in $\{0, 1\}$, each with probability $\frac{1}{2}$.
    
    At each time step $i \in \{1,\dotsc,n\}$, the {\em teacher} observes $\Theta$ through BSC($p$), a binary symmetric channel with error rate $p \in \big(0,1/2\big)$:
    \begin{equation}
        \p(Y_i=\Theta) = 1-p, ~~~ \p(Y_i=1-\Theta) = p.
    \end{equation}
    The binary symmetric channel is \textit{memoryless}, meaning that the $Y_i$'s are conditionally independent given $\Theta$.  In Section \ref{sec:dmc}, we will study more general and possibly non-symmetric binary-input discrete memoryless channels.
    
    At time $i$, the teacher then transmits binary information $\hat{X}_i$ to the {\em student} through another binary symmetric channel BSC($q$), where $q \in \big(0,1/2\big)$. Thus, denoting the student's observations by $Z_i$, we have 
    \begin{equation}
        \p(\hat{X}_i = Z_i) = 1-q, ~~~ \p(\hat{X}_i = 1-Z_i) = q,
    \end{equation}
    and the $Z_i$ are conditionally independent given $\hat{X}_i$.  Importantly, $\hat{X}_i$ must only be a function of $Y_1, ,\ldots, Y_i$; the teacher cannot look into the future.
    
    At time $n$, having received $Z_1, \ldots, Z_n$, the student makes an estimate of $\Theta$, which we denote by $\hat{\Theta}_n$ (or sometimes simply $\hat{\Theta}$).  Then, the learning rate is defined as 
    \begin{equation}
        \mathcal{R} = \limsup_{n\rightarrow \infty} \left\{ -\frac1n \ln \mathbb{P}(\hat{\Theta}_n \neq \Theta)\right\}.
    \end{equation}
    We assume that the noise parameters $p$ and $q$ are known to both the teacher and student; removing this assumption may be an interesting direction for future work.
    
    \subsection{Existing Results}
    
    The preceding motivation and setup follows that of Jog and Loh \cite{jog2020teaching}, and the same problem was also considered with different motivation and terminology by Huleihel, Polyanskiy, and Shayevitz \cite{onebit}, who framed the problem using the terminology of relay channels.  While relay channels are already a central topic in information theory, the authors of \cite{onebit} further motivated their study via the more recent \emph{information velocity} problem, which was posed by Yury Polyanskiy \cite{onebit} and is also captured under a general framework studied by Rajagopalan and Schulman \cite{schulman_1994}.  In that problem, the goal is to reliably transmit a single bit over a long chain of relays while maintaining a non-vanishing ratio between the number of hops and the total transmission time.  Below we will discuss a conjecture (for the 2-hop setting) made in \cite{onebit} that is directly inspired by the connection with information velocity.
    
    It is known that the student's optimal strategy is the maximum likelihood decoder (assuming the teacher's strategy is known) \cite{onebit,jog2020teaching}.  In \cite{jog2020teaching}, it was discussed that optimal strategies can be difficult to analyze in general, and accordingly, the following two simpler student strategies were considered:
    \begin{itemize}
        \item {\em Majority rule:} $\Theta_n$ is chosen according to the majority of the observations $Z_1, \ldots, Z_n$.
        \item {\em $\epsilon$-majority rule:} $\Theta_n$ is  the majority of the last $\epsilon\cdot n$ observations for some $\epsilon \in (0,1)$, with the rough idea being to allow time for the teacher to learn $\Theta$ first.
    \end{itemize}
    These student strategies were analyzed together with the following teacher strategies:	
    \begin{itemize}
        \item {\em Simple forwarding:} The teacher directly forwards its observation at each time step, i.e., $\hat{X}_i = Y_i$.
        \item {\em Cumulative teaching:} The teacher transmits its best estimate of $\Theta$ at each time step, i.e., $\hat{X}_i$ is the majority value among $Y_1, \ldots, Y_i$. 
    \end{itemize}
    It was demonstrated in \cite{jog2020teaching} that none of these strategies uniformly outperform one an another, and that any combination of them falls short of the $D(1/2\|\max(p,q))$ upper bound (i.e., converse) based on data processing arguments (here and subsequently, $D(p_1\|p_2)$ denotes the relative entropy between Bernoulli distributions with parameters $p_1$ an $p_2$).  The optimal learning rate was left as an open problem, and it was conjectured that at least part of the weakness is in the upper bound, i.e., that one cannot attain a rate of $D(1/2\|\max(p,q))$.
    
    In \cite{onebit}, both simple forwarding and cumulative teaching (i.e., relaying) were considered, along with a block-structured variant that only updates the majority every $\sqrt{n}$ time steps.  These strategies were analyzed alongside the optimal maximum-likelihood learning (i.e., decoding) rule.  While the learning rates attained again fall short of the $D(1/2\|\max(p,q))$ upper bound, it was shown that one comes within a factor $\frac{3}{4}$ when $p=q$ and $p \to \frac{1}{2}$ (i.e., the high noise setting).  It was conjectured in \cite{onebit} that using a more sophisticated protocol, this constant $\frac{3}{4}$ could be made arbitrarily close to one in this high-noise limit.  As hinted above, the motivation for this conjecture was the fact that positive information velocity is known to be attainable; the authors of \cite{onebit} discuss that for this to be possible, the 1-hop and 2-hop exponents should match in the high-noise limit so that ``information propagation does not slow down''.
    
    The study of \cite{jog2020teaching} was inspired by earlier works on social learning \cite{vives1993fast,jadbabaie2013information,molavi2017foundations}, and most similar to \cite{harel2020}.  Since these are less directly relevant to our work, we refer the reader to \cite{jog2020teaching} for a detailed discussion of the differences.
    
    As we will see shortly, our main result disproves the above-mentioned conjecture of \cite{jog2020teaching}, and not only confirms the conjecture of \cite{onebit} for the regime $p \to \frac{1}{2}$, but also significantly strengthens it by extending it to all $p \in \big(0,\frac{1}{2}\big)$.  
    
    %
    %
    %
    
    \section{Optimal Rate Under Binary Symmetric Noise} \label{sec:bsc}
    
    In this section, we focus on the above-described setup in which both channels in the system are binary symmetric channels.  We turn to more general binary-input discrete memoryless channels in Section \ref{sec:dmc}.
    
    \subsection{Statement of Optimal Learning Rate}
    
    Our first main result is stated as follows.
    
    \begin{theorem} \label{thm:main}
        Under the preceding setup, for any $p,q \in \big(0,1/2\big)$, let $\mathcal{R}^*(p,q)$ be the supremum of learning rates across all teaching and learning protocols. Then
        \begin{equation}
            \mathcal{R}^*(p,q) \ge D\Big(\frac{1}{2} \,\Big\| \max(p,q)\Big),
        \end{equation}
        where $D(a\|b) = a \ln \frac{a}{b} + (1-a) \ln\frac{1-a}{1-b}$ denotes the binary relative entropy function (in base $e$).
    \end{theorem}
    
    We achieve this result via a novel block-structured teaching strategy.  In contrast with \cite{onebit}, the teacher transmits bits based on a single block at a time, ignoring earlier blocks.  The teacher encodes the number of 1s received in the block by sending a sorted sequence (i.e., a string of $1$s followed by a string of $0$s), with the transition point chosen according to a carefully-chosen non-linear function of the number of 1s observed.  The student performs maximum-likelihood decoding, and we study the error probability via the Bhattacharyya coefficient.  The details are given in Sections \ref{sec:protocol} and \ref{sec:error}, and in Section \ref{sec:computaiton} we show that the overall computation time required by the teacher and student is linear in $n$.
    
    {\bf Comparison to upper bound (converse).} If $q=0$, then the optimal learning rate is given by the majority decoder, which achieves a learning rate of \cite{jog2020teaching,onebit}
    \begin{equation}
        D\Big(\frac{1}{2} \,\Big\| p\Big) = \frac{1}{2} \ln \left(\frac{1}{4p(1-p)}\right). \label{eq:1hop_bsc}
    \end{equation}
    Through data processing inequalities, we have $\mathcal{R}^*(p,q) \leq \mathcal{R}^*(p,0) = D\big(\frac{1}{2} \| p\big)$, and an analogous argument yields $\mathcal{R}^*(p,q) \leq D\big(\frac{1}{2} \| q\big)$. We therefore obtain the upper bound
    \begin{equation}
        \mathcal{R}^*(p,q) \leq D\Big(\frac{1}{2} \,\Big\| \max(p,q)\Big).
        \label{optimal}
    \end{equation}
    Theorem \ref{thm:main} shows that this bound is tight, disproving the above-mentioned conjecture of \cite{jog2020teaching}, and both confirming and strengthening the conjecture of \cite{onebit}.
    
    In the remainder of the section, we only consider the case $p=q$.  This is without loss of generality, since in the general case, one can simulate additional noise into the system to make the noise parameters both equal to $\max(p,q)$.
    
    \subsection{Teacher and Student Protocol} \label{sec:protocol}
    
    Before describing the protocol, we first need to introduce some statistical tools.
    
    \subsubsection{Statistical Tools and Student Decoding Rule}
    
    Let $A$ and $B$ be discrete random variables defined on some finite alphabet $\Omega$. The Bhattacharyya coefficient is a real number in $[0,1]$ given by
    \begin{equation}
        \rho(A,B) = \sum_{x \in \Omega} \sqrt{\p(A=x)\p(B=x)}.
        \label{bc-def}
    \end{equation}
    Overloading notation, we will sometimes also denote this quantity by $\rho(P_A,P_B)$, where $P_A$ and $P_B$ are the associated probability mass functions.
    
    The Bhattacharyya coefficient measures the `closeness' between the distributions of $A$ and $B$. Values close to 0 indicate easily separated distributions, while values close to 1 indicate very similar distributions. In particular, $\rho(A,B)=0$ if and only if $A$ and $B$ have disjoint supports, and $\rho(A,B)=1$ if and only if they are identically distributed.
    
    We will make use of some standard properties of the Bhattacharyya coefficient:
    \begin{itemize}
        \item Suppose that we have a random variable $X$ which is known to follow one of the two known distributions, $A$ or $B$. If we draw one instance of $X$ and use it to decide which distribution $X$ came from, then there exists a strategy to achieve an error probability of at most $\rho(A,B)$.
        To achieve this, we simply use the maximum likelihood test: Given $X=x$, we decide that $A$ is the true distribution if $\p(A=x) > \p(B=x)$, and decide that $B$ is the true distribution otherwise.
        The probability of selecting distribution $A$ when $B$ is indeed the true distribution is given by
        \begin{align}
            &\sum_{x \in \Omega} \p(B=x) \mathbbm{1}_{\p(A=x)>\p(B=x)} \nonumber \\ &\qquad
            \leq \sum_{x \in \Omega} \sqrt{\p(A=x) \p(B=x)}.
        \end{align}
        The other error type is upper bounded similarly, and combining the two error types yields the desired bound:
        \begin{equation}
            \p({\rm error}) \le \rho(A,B). \label{error}
        \end{equation}
        
        \item If $A_1, A_2$ are independent, and similarly for $B_1, B_2$, then
        \begin{equation}
            \rho((A_1, A_2), (B_1, B_2)) = \rho(A_1, B_1) \cdot \rho(A_2, B_2).
            \label{mult}
        \end{equation}
        This follows by a direct expansion of (\ref{bc-def}).
        \item Let $\vec{A} = (A_1, A_2, \ldots, A_m)$ be $m$ i.i.d. copies of $A$ and $\vec{B} = {B_1, B_2,\ldots, B_m}$ be $m$ i.i.d. copies of $B$. Then
        \begin{equation}
            \rho(\vec{A}, \vec{B}) = \rho(A,B)^m.
            \label{power}
        \end{equation}
        This is a direct consequence of (\ref{mult}).
    \end{itemize}
    In accordance with \eqref{error}, we assume throughout the paper that the student adopts the optimal maximum-likelihood decoding rule, where the two underlying distributions are those of $\Theta = 1$ vs.~$\Theta = 0$ given $Z_1,\dotsc,Z_n$.
    
    \subsubsection{Block-Structured Design and Teaching Strategy}
    
    The transmission protocol of length $n$ is broken down into $n/k$ blocks, each of length $k$.  For each $i=1,2,3,\ldots, \frac{n}{k}-1$, we let $(\hat X_{ik+1}, \ldots, \hat X_{(i+1)k})$ be a (deterministic) function of $(Y_{(i-1)k+1}, \dotsc, Y_{ik})$. The values of $\hat{X}_1, \ldots, \hat{X}_k$ are ignored by the student. The function used to generate $\vec{\hat{X}} \in \{0,1\}^{k}$ from $\vec{Y} \in \{0,1\}^{k}$ is given in Section \ref{within-block}.
    
	Under this design, the $i$-th block of $\hat{X}$ only depends on the $(i-1)$-th block of $Y$; see Figure \ref{block-protocol} for an illustration.    To simplify notation, we define
    \begin{equation}
        W_i = (Z_{ik+1},, \ldots, Z_{(i+1)k}) \label{eq:W}
    \end{equation}
    to be the $i$-th block received by student. 
    
    \begin{figure}
        \centering
        \scalebox{0.85}{
            \begin{tikzpicture}
                [line cap=round,line join=round,>=triangle 45,x=1.0cm,y=1.0cm, scale=0.6]
                \clip(-1,-1) rectangle (16, 7);
                \draw [line width=1.pt] (0.,0.)-- (12.,0.);
                \draw [line width=1.pt] (0.,1.)-- (12.,1.);
                \draw [line width=1.pt] (0.,3.)-- (12.,3.);
                \draw [line width=1.pt] (12.,1.)-- (12.,0.);
                \draw [line width=1.pt] (0.,1.)-- (0.,0.);
                \draw [line width=1.pt] (4.,1.)-- (4.,0.);
                \draw [line width=1.pt] (8.,1.)-- (8.,0.);
                \draw [line width=1.pt] (0.,4.)-- (12.,4.);
                \draw [line width=1.pt] (12.,3.)-- (12.,4.);
                \draw [line width=1.pt] (8.,3.)-- (8.,4.);
                \draw [line width=1.pt] (4.,3.)-- (4.,4.);
                \draw [line width=1.pt] (0.,3.)-- (0.,4.);
                \draw [->,line width=1.pt] (0,5) -- (13,5);
                \draw (14,5) node{time};
                \draw (0, 4.5) node{1};
                \draw (3.8, 4.5) node{$k$};
                \draw (7.7, 4.5) node{$2k$};
                \draw (11.7, 4.5) node{$3k$};
                \draw (0, -0.5) node{1};
                \draw (3.8, -0.5) node{$k$};
                \draw (7.7, -0.5) node{$2k$};
                \draw (11.7, -0.5) node{$3k$};
                \draw [->,line width=1.pt] (2,3) -- (6,1);
                \draw [->,line width=1.pt] (6,3) -- (10,1);
                \draw [fill=red] (0,3) rectangle (4,4);
                \draw [fill=red] (4,0) rectangle (8,1);
                \draw [fill=blue] (4,3) rectangle (8,4);
                \draw [fill=blue] (8,0) rectangle (12,1);
                \draw (14,3.5) node{teacher};
                \draw (14,0.5) node{student};
            \end{tikzpicture}
        }
        \caption{A diagrammatic representation of the block protocol}
        \label{block-protocol}
    \end{figure}
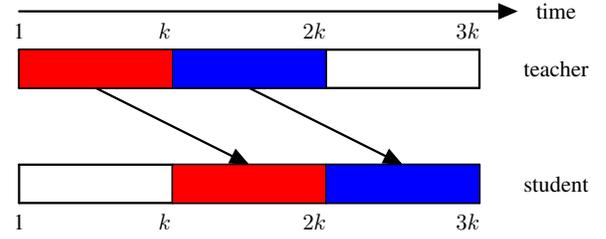

    The memoryless properties of the channels imply that the $W_i$'s are conditionally independent given $\Theta$. Since the function that generates $\vec{\hat{X}}$ from $\vec{Y}$ is the same every time, their distributions are also identical; we denote this common distribution by $P_W$.
    
    The student needs to distinguish between $n/k-1$ i.i.d. copies of $P_{W|\Theta=1}$ and $n/k-1$ i.i.d. copies of $P_{W|\Theta=0}$. Using (\ref{error}) and (\ref{power}), we have
    \begin{equation}
        \p(\hat\Theta_n \neq \Theta) \leq \rho(P_{W|\Theta=1},P_{W|\Theta=0})^{n/k-1}.
        \label{rate-k}
    \end{equation}
    By setting $k$ to be a fixed constant as $n$ increases, it follows that the learning rate is lower bounded by
    \begin{align}
        \mathcal{R} &= \limsup_{n\rightarrow \infty} \left\{ -\frac1n \ln \mathbb{P}(\hat{\Theta}_n \neq \Theta)\right\} \\
        &\geq -\frac1k \ln \rho(P_{W|\Theta=1},P_{W|\Theta=0}),
        \label{rate-2}
    \end{align}
    where $W$ implicitly depends on $k$.   In fact, while constant $k$ suffices for our purposes, a similar argument holds when $k$ increases but satisfies $k = o(n)$, in which case the operation $\limsup_{k\rightarrow \infty}$ should also be included on the right-hand side of \eqref{rate-2}.
    
    \begin{remark}
        This block structure bears high-level similarity to block-Markov coding in the analysis of relay channels \cite[Sec.~16.4]{el2011network}.  However, the details here are very different from existing analyses that use joint typicality and related notions.  The error exponent associated such analyses would be low due to the effective block length of $k$, whereas we maintain a high error exponent via joint decoding over the entire length-$n$ received sequence.
    \end{remark}
    
    \subsubsection{Protocol within Each Block}
    \label{within-block}
    In each block, the teacher receives $k$ noisy bits (e.g., $Y_1,\dotsc,Y_k$ in the first block). Let $\alpha$ be the fraction of 1 among these bits (i.e. the number of 1s divided by $k$).
    
    The teacher sends $k\cdot f(\alpha)$ bits of 1 followed by $k(1-f(\alpha))$ bits of $0$ (with a delay of one block), where $f: [0,1] \rightarrow [0,1]$ is defined as
    \begin{equation}
        f(\alpha) = 
        \begin{cases}
            0 & 0\leq \alpha < p\\
            \frac{D(\alpha\|p)}{2D(1/2\|p)} & p\leq \alpha \le 1/2\\
            1-f(1-\alpha) & 1/2 < \alpha < 1-p\\
            1 & 1-p \leq \alpha \leq 1.
        \end{cases}
    \end{equation}
    A sample plot is given in Figure \ref{plot} with $p=0.2$.  As exemplified in the figure, $f$ is non-decreasing, and
    \begin{gather}
        f(\alpha) = 1-f(1-\alpha),
        \label{symmetric} \\
        f(\alpha) \leq \frac{D(\alpha\|p)}{2D(1/2\|p)},
        \label{f-upper-bound}
    \end{gather}
    where \eqref{f-upper-bound} is trivial for $\alpha \le \frac{1}{2}$, and follows easily from the convexity of $D(\cdot \| p)$ for $\alpha > \frac{1}{2}$.  Specifically, this convexity implies that the derivative of $D(\cdot \| p)$ increases for $\alpha > \frac{1}{2}$, whereas the chain rule applied to \eqref{symmetric} implies that the derivative of $f(\alpha)$ decreases for $\alpha > \frac{1}{2}$.   Since \eqref{f-upper-bound} holds with equality for $\alpha \in \big[p,\frac{1}{2}\big]$ by definition, this behavior of the derivatives implies that \eqref{f-upper-bound} also holds for $\alpha > \frac{1}{2}$.
    
    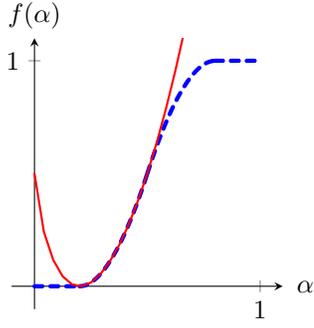
\begin{figure}
        \begin{center}
            \begin{tikzpicture}[line cap=round,line join=round,>=triangle 45,x=3.0cm,y=3.0cm]
                \draw (1.3,0.1) node{$\alpha$};
                \draw (0.1,1.3) node{$f(\alpha)$};

                \begin{axis}[
                    x=3.0cm,y=3.0cm,
                    axis lines=middle,
                    xmin=-0.1,
                    xmax=1.1,
                    ymin=-0.1,
                    ymax=1.1,
                    xtick={0,1},
                    ytick={0,1},]
                    \clip(-0.1,-0.1) rectangle (1.1,1.1);
                    \addplot[ultra thick,domain=0:0.2, color=blue, dashed] {0};
                    \addplot[ultra thick,domain=0.2:0.5, color=blue, dashed] {(x*log10(x/0.2)+(1-x)*log10((1-x)/0.8))/0.1938};
                    \addplot[ultra thick,domain=0.5:0.8, color=blue,dashed] {1-((1-x)*log10((1-x)/0.2)+x*log10(x/0.8))/0.1938};
                    \addplot[ultra thick,domain=0.8:1, color=blue,dashed] {1};
                    \addplot[thick, domain=0:1, color=red] {(x*log10(x/0.2)+(1-x)*log10((1-x)/0.8))/0.1938};
                \end{axis}
            \end{tikzpicture}
            \caption{The function $f$ used in the protocol, plotted for $p=0.2$.  The dashed line represents $f(\alpha)$, while the solid line represents $\frac{D(\alpha\|p)}{2D(1/2\|p)}$.}
            \label{plot}
        \end{center}
    \end{figure}
    
    In general $k\cdot f(\alpha)$ may not be an integer. When this happens, we can use $\lfloor k\cdot f(\alpha) \rfloor$ bits of $1$ followed by $k-\lfloor k\cdot f(\alpha) \rfloor$ bits of $0$. The rounding does not affect the error analysis, so to simplify notation, we focus on the integer-valued case.  See Footnote \ref{foot:rounding} on Page \pageref{foot:rounding} for an additional remark on this issue.
    
    Before proceeding with the student strategy and error analysis, we intuitively explain why sending a sorted sequence (i.e. with all the 1s sent before the $0$s) can make decoding easier for the student.  Suppose the student receives the string 1111000000\underline{1}0. If the student knows in advance that the string sent by the teacher is sorted, they can be confident that the underlined bit is flipped. However, if this were to be an unsorted string, such corrections could not be made.
    
    The reasoning behind our choice of $f$ will become apparent in the error analysis in the following subsection, but intuitively, it serves as a carefully-chosen middle ground between the simpler choices $f(\alpha) = \alpha$ (bearing some resemblance to simple forwarding) and $f(\alpha) = \boldsymbol{1}\{ \alpha > \frac{1}{2} \}$ (bearing some resemblance to cumulative teaching).
    
    \subsection{Error Analysis} \label{sec:error}
    
    Continuing the study of a single block $W$ defined according to \eqref{eq:W}, observe that we have the Markov chain
    \begin{equation}
        \Theta \rightarrow \alpha \rightarrow W,
    \end{equation}
    and in accordance with \eqref{rate-2}, we would like the distributions $P_{W|\Theta=1}$ and $P_{W|\Theta=0}$ to be as `far apart' as possible.
    %
    We break the error analysis into two parts, analyzing the relations $\Theta \rightarrow \alpha$ and $\alpha \rightarrow W$ separately.
    
    \subsubsection{Relating $\Theta$ and $\alpha$}
    Define $W_{\alpha}$ to follow the conditional distribution of $W$ given $\alpha$.
    The distribution of $W_{\alpha}$ is simple; recalling that we are focusing on the case $p = q$, we have the following:
    \begin{itemize}
        \item The bits of $W_{\alpha}$ are independently distributed;
        \item If $i \le k\cdot f(\alpha)$, then the $i$-th bit is 1 with probability $1-p$ and $0$ otherwise;
        \item If $i > k\cdot f(\alpha)$, then the $i$-th bit is 1 with probability $p$ and $0$ otherwise.
    \end{itemize}
    
    We proceed by upper bounding $\rho(W_{\alpha_1}, W_{\alpha_2})$, starting with the case that $\alpha_1 \leq \alpha_2$.
    
    The bits up to index $k\cdot f(\alpha_1)$, and the bits after index $k\cdot f(\alpha_2)$, are all identically distributed between $W_{\alpha_1}$ and $W_{\alpha_2}$, and thus do not provide any distinguishing power. We are left with $k\cdot f(\alpha_2)-k\cdot f(\alpha_1)$ bits, which are i.i.d.~according to either Bernoulli($p$) or Bernoulli($1-p$).
    
    The \bh coefficient between Bernoulli($p$) and Bernoulli($1-p$) can easily be computed to be $2\sqrt{p(1-p)} = e^{-D(1/2\|p)}$,
    and using (\ref{power}), we obtain\footnote{\label{foot:rounding}If we were to explicitly incorporate rounding as discussed following \eqref{f-upper-bound}, then this equation would be replaced by the slightly looser bound $\rho(W_{\alpha_1}, W_{\alpha_2}) \le e^{-D(1/2\|p)\cdot (k f(\alpha_2) - kf(\alpha_1) - 1)}$.  The subtraction of one in the exponent only amounts to multiplying the entire expression by a constant, which does not impact the resulting exponential decay rate.}
    \begin{equation}
        \rho(W_{\alpha_1}, W_{\alpha_2}) = e^{-kD(1/2\|p)\cdot (f(\alpha_2) - f(\alpha_1))}
        \label{bc-alpha-tmp}
    \end{equation}
    To simplify the exponent, we use (\ref{symmetric}) and (\ref{f-upper-bound}) to obtain 
    \begin{gather}
        f(\alpha_1) \leq \frac{D(\alpha_1\|p)}{2D(1/2\|p)}, \\
        f(\alpha_2) = 1-f(1-\alpha_2) \geq 1-\frac{D(1-\alpha_2\|p)}{2D(1/2\|p)},
    \end{gather}
    and hence,
    \begin{equation}
        f(\alpha_2)-f(\alpha_1) \geq 1 -\frac{D(\alpha_2\|p)}{2D(1/2\|p)} - \frac{D(1-\alpha_1\|p)}{2D(1/2\|p)}.
    \end{equation}
    We can therefore weaken (\ref{bc-alpha-tmp}) to
    \begin{equation}
        \rho(W_{\alpha_1}, W_{\alpha_2}) \leq e^{-k \cdot (D(1/2\|p) - D(\alpha_1\|p)/2 - D(1-\alpha_2\|p)/2)}.
        \label{bc-alpha}
    \end{equation}
    Although the assumption $\alpha_1 \leq \alpha_2$ was used in the derivation of (\ref{bc-alpha-tmp}), a trivial argument shows that when $\alpha_1 > \alpha_2$, the right-hand side of (\ref{bc-alpha-tmp}) is at least one (since $f$ is non-decreasing). Since the Bhattacharyya coefficient never exceeds one, it follows that (\ref{bc-alpha}) remains true in the case that $\alpha_1 > \alpha_2$.
    
    \subsubsection{Relating $\alpha$ and $W$}
    
    The next step is to decompose the distribution of $W$ into the simpler distributions $W_\alpha$. To do so, we use the following definition.
    \begin{dfn} \label{def:mixture}
        Let $A_1, \ldots, A_m$ be probability distributions defined on a common finite probability space $\Omega$ and $p_1, \ldots, p_m$ be real numbers in $[0,1]$ that sum to one. The mixture distribution $A$ described by $[p_1,A_1; \ldots; p_k,A_k]$ is defined as $\p(A=x) = \sum_{i=1}^k p_i \cdot \p(A_i=x)$.
    \end{dfn}
    
    We proceed with a simple technical lemma on the \bh coefficient on mixtures.
    
    \begin{lemma} \label{lem:mixture}
        Suppose that $A$ follows a mixture distribution described by $[p_1,A_1; \ldots; p_k,A_k]$. Then for any distribution $B$, we have
        \begin{equation}
            \rho(A,B) \leq \sum_{i=1}^k \sqrt{p_i} \rho(A_i, B).
            \label{mixture}
        \end{equation}
    \end{lemma}
    \begin{proof}
        Using the inequality $\sqrt{a+b} \leq \sqrt{a}+\sqrt{b}$, we have
        \begin{align}
            \rho(A,B) &= \sum_{x \in \Omega} \sqrt{\p(A=x)\p(B=x)} \\
            &= \sum_{x\in \Omega} \sqrt{\sum_{i=1}^k p_i\p(A_i=x) \p(B=x)} \\
            &\leq \sum_{x\in \Omega} \sum_{i=1}^k \sqrt{p_i\p(A_i=x) \p(B=x)}.
        \end{align}
	   Simple re-arranging gives the desired result.
    \end{proof}
    
    Using the symmetry of the \bh coefficient and applying Lemma \ref{lem:mixture} twice, we obtain
    \begin{multline}
        \rho(P_{W|\Theta=1},P_{W|\Theta=0})
        \leq \sum_{i=0}^k\sum_{j=0}^k \sqrt{\p(\alpha=i/k|\Theta=1)}  \\ \times
        \sqrt{\p(\alpha=j/k|\Theta=0)} \cdot \rho(W_{i/k}, W_{j/k}).
        \label{decompose}
    \end{multline}
    Moreover, the conditional distribution of $\alpha$ given $\Theta$ is simple:
    \begin{align}
        (\alpha|\Theta=1) &\sim \frac1k\cdot \text{Binomial}(k, 1-p), \\
        (\alpha|\Theta=0) &\sim \frac1k\cdot \text{Binomial}(k, p).
    \end{align}
    Hence, applying the Chernoff bound for the binomial distribution (e.g., see \cite[Sec.~2.2]{Bou13}), we obtain
    \begin{equation}
        \p(\alpha=\alpha_0|\Theta=1) \leq e^{-kD(\alpha_0\|1-p)} = e^{-kD(1-\alpha_0\|p)},
    \end{equation}
    and similarly, 
    \begin{equation}
        \p(\alpha=\alpha_0|\Theta=0) \leq e^{-kD(\alpha_0\|p)}.
    \end{equation}
    
    Combining these bounds with (\ref{bc-alpha}) and (\ref{decompose}), we find that the $D(i/k\|p)/2$ and $D(1-j/k\|p)/2$ terms both cancel to zero, and we are left with
    \begin{align}
        \rho(P_{W|\Theta=1},P_{W|\Theta=0}) &\le \sum_{i=0}^k \sum_{j=0}^k e^{-kD(1/2\|p)} \\
        &= (k+1)^2 e^{-kD(1/2\|p)}.
        \label{bc-final}
    \end{align}
    Combining (\ref{rate-2}) with (\ref{bc-final}) gives
    \begin{equation}
        \mathcal{R} \geq D(1/2\|p) - \mathcal{O}\left(\frac{\log k}{k}\right),
    \end{equation}
    and hence, by setting $k$ large enough, we can obtain any learning rate arbitrarily close to $D(1/2\|p)$.  This completes the proof of Theorem \ref{thm:main}.
    
    \subsection{Computational complexity} \label{sec:computaiton}
    
    In this section, we argue that it is possible to execute the described protocol with an overall computation time of $\mathcal{O}(n)$ (i.e., no higher than that of simply reading the received bits).  This is immediate for the teacher, so we focus our attention on the student.
    
    To implement the maximum likelihood decoder, it suffices to compute the associated log likelihood ratio (LLR) $\log \frac{\p(\vec{Z} =\vec{z} | \Theta = 1)}{\p(\vec{Z} =\vec{z} | \Theta = 0)}$ given the $n$ received bits $\vec{z}$.  In addition, since the $n/k-1$ blocks are independent by construction and memorylessness, the overall LLR is the sum of per-block LLRs. Hence, we only need to show that each of these can be computed in time $\mathcal{O}(k)$.
    
    The only minor difficulty in computing the LLR for a single block is that we need to account for all possible $k+1$ choices of the latent/hidden variable $\alpha \in \big\{0,\frac{1}{k},\dotsc,1\big\}$: Letting $W$ denote a generic length-$k$ block according to \eqref{eq:W}, we have
    \begin{equation}
        \p(W | \Theta = 1) = \sum_{\alpha} \p(\alpha | \Theta = 1) \p(W|\alpha), \label{eq:sum}
    \end{equation}
    and similarly for $\p(W | \Theta = 0)$.  In addition, since the channel is memoryless, we have
    \begin{equation}
        \p(W|\alpha) = \prod_{i=1}^k \p(W(i)|\alpha), \label{eq:W_prod}
    \end{equation}
    where $W(i)$ is the $i$-th bit comprising $W$.  Naively, we could evaluate \eqref{eq:sum} and \eqref{eq:W_prod} directly, but this would give a per-block complexity of $\mathcal{O}(k^2)$.
    
    To improve this, we recall that for two values $\alpha_1,\alpha_2 \in \big\{0,\frac{1}{k},\dotsc,1\big\}$, the associated values of $\p(W(i)|\alpha)$ only differ for values of $i$ in between $k\cdot f(\alpha_1)$ and $k\cdot f(\alpha_2)$.  Therefore, we can initially compute $\p(W|\alpha)$ for $\alpha = 0$, and then use it to compute the value for $\alpha = \frac{1}{k}$ by only looking at the first $k\cdot f\big(\frac{1}{k}\big)$ bits of $W$, and so on for $\alpha \in \big\{\frac{2}{k},\frac{3}{k},\dotsc,1\big\}$; at each step, we only need to look at the bits affected by incrementing $\alpha$.  In this manner, it only takes $\mathcal{O}(k)$ time to compute {\em all} $k+1$ values of $\p(W|\alpha)$.  In addition, $\p(\alpha | \Theta = 1)$ in \eqref{eq:sum} follows a scaled binomial distribution, whose probability mass function can be pre-computed in $\mathcal{O}(k)$ time (assuming constant-time arithmetic operations).  Hence, we obtain an overall per-block computation time of $\mathcal{O}(k)$.
    
    In addition to the low computational overhead, another advantage of our strategy is that it is \textit{anytime}.  Specifically, it can be run without knowledge of $n$, and we can stop the algorithm any time and ask the student for the estimate $\hat{\Theta}_n$.  Letting $\epsilon_{k,n}$ be the resulting error probability with block size $k$ and total time $n$, it still holds in this scenario that  $-\frac{1}{n} \ln \epsilon_{k,n} \rightarrow D(1/2||p)$, as long as $k \rightarrow \infty$ and $n/k \rightarrow \infty$.
    
    \section{General Binary-Input Discrete Memoryless Channels}  \label{sec:dmc}
    
    In this section, we consider the same model as Section \ref{model-dfn}, except that we replace BSC($p$) and BSC($q$) by general discrete memoryless channels (DMCs) $P$ and $Q$ (known to the student and teacher), both of which have binary inputs $\{0,1\}$. We denote the transition probabilities by $P(y|\theta)$ and $Q(z|\hat{x})$.  The case that $Q$ has non-binary input is discussed in Section \ref{sec:conclusion}.
    
    We first overview the 1-hop case in Section \ref{sec:1hop}, following the classical work of Shannon, Gallager, and Berlekamp \cite{shannon1967}.  We then state an achievable learning rate in Section \ref{sec:dmc_ach}, and give the protocol and analysis in Section \ref{sec:dmc_protocol}, which are essentially more technical variants of the binary symmetric case.  We provide an upper bound on the learning rate (i.e., a  converse) in Section \ref{sec:dmc_conv}, and use it to deduce the optimal learning rate in certain special cases.  While the tightness of our achievable learning rate remains open in general, we additionally show in Section \ref{sec:block_conv} that, at least under certain technical assumptions, it cannot be improved for block-structured protocols.
    
    \subsection{Existing Results on the 1-Hop Case} \label{sec:1hop}
    
    We first consider a 1-hop scenario in which a single agent makes repeated observations of the unknown variable $\Theta$ through repeated uses of a binary-input DMC $P$. This is a well-studied problem, and we will summarize some of the results from \cite{shannon1967}.
    
    We begin with a generalization of the \bh coefficient. For a real number $s \in (0,1)$ and two random variables $A, B$ defined on the same finite alphabet $\Omega$, let
    
    \begin{equation}
        \rho(A, B, s) = \sum_{x \in \Omega} \p(A=x)^{1-s} \p(B=x)^s \in [0,1].
    \end{equation}
    where the upper bound of one follows by applying Jensen's inequality to $\mathbb{E}_A\big[ \big( \frac{P_B(B)}{P_A(A)} \big)^s \big]$.  This function can be extended to $s=0$ and $s=1$ by continuity; for instance, $\rho(A, B, 0) = \sum_{x \in \Omega} \p(A=x) \mathbbm{1}_{\p(B=x)> 0}$.  We also note that setting $s = \frac{1}{2}$ recovers the \bh coefficient.  We write $\rho(A,B,s)$ and $\rho(P_A,P_B,s)$ interchangeably.
    Similarly to \eqref{mult}--\eqref{power}, we have the following properties for independent pairs and length-$m$ i.i.d.~vectors:
    \begin{gather}
        \rho((A_1, A_2), (B_1, B_2),s) = \rho(A_1, B_1,s) \cdot \rho(A_2, B_2,s),
        \label{mult2} \\
        \rho(\vec{A}, \vec{B},s) = \rho(A,B,s)^m.
        \label{power2}
    \end{gather}
    
    Given the DMC $P$ with inputs $0$ and $1$, we define
    \begin{equation}
        \rho_P(s) = \rho(P(\cdot|0), P(\cdot|1), s),
    \end{equation}
    and denote its logarithm by\footnote{See Figure \ref{fig:mu_examples} on Page \pageref{fig:mu_examples} for some examples of how $\mu_P(s)$ varies as a function of $s$.}
    \begin{equation}
        \mu_P(s) = \ln \rho_P(s) \le 0,
    \end{equation}
    where $P(\cdot | 0)$ and $P(\cdot |1)$ are treated as probability distributions on the output alphabet.  The first and second derivatives of $\mu_P(s)$ for $s \in (0,1)$ are denoted by $\mu'_P(s)$ and $\mu''_P(s)$, and once again, the values for the endpoints $s \in \{0,1\}$ are defined with respect to the appropriate limits.  
    
    Suppose that the agent makes a single observation of $\Theta \in \{0,1\}$ through $P$, and uses it to estimate $\Theta$.  Using the identity $\mathbbm{1}_{\p(B=x)>\p(A=x)} \leq \big(\frac{\p(B=x)}{\p(A=x)}\big)^s$, we find that for any $s \ge 0$, $\rho_P(s)$ is an upper bound to the probability of a maximum-likelihood decision rule selecting $B$ when $A$ is true.  An analogous property holds for the reverse error event with  $1-s$ in place of $s$, and it follows that $\rho_P(s)$ is an upper bound to the overall error probability for all $s \in [0,1]$.
    
    In view of the property in \eqref{power2}, if the agent instead makes $n$ independent observations, the upper bound becomes $\rho_{P^n}(s) = \rho_P(s)^n$, implying a learning rate of at least
    \begin{equation}
        -\frac1n \ln \rho_P(s)^n = -\mu_P(s).
        \label{mu-learning-rate}
    \end{equation}
    The following well-known result reveals that this learning rate is optimal upon optimizing over $s \in [0,1]$.
    
    \begin{lemma} \label{lem:berlekamp}
        {\em \cite[Corollary of Thm.~5]{shannon1967}}\footnote{This is a weakened version of the statement in \cite{shannon1967}, in which we apply $\max\{s,1-s\} \le 1$ in the coefficient to the $\sqrt{n}$ term.  In addition, we specialize to the case that every transmitted symbol is identical, in order to match our setup. } Let $s^* \in [0,1]$ be chosen to minimize $\mu_P(s)$. Then any decoding strategy of the agent after $n$ uses of the channel must have an error probability of at least $\frac{1}{8} \exp\big(n\mu_P(s^*) - \sqrt{2n\mu_P''(s^*)}\big)$ when $\Theta$ is equiprobable on $\{0,1\}$.
    \end{lemma}
    
    Since the $\sqrt{2n\mu_P''(s^*)}$ term is sub-linear in $n$, this immediately gives the following corollary.
    \begin{cor}
        The optimal learning rate for the 1-hop system is given by $-\min_{s \in [0,1]} \mu_P(s)$.
        \label{1-hop-error}
    \end{cor}
    
    \subsection{Achievability Result} \label{sec:dmc_ach}
    
    We now state the second main result of our paper, giving an achievability result for general binary-input DMCs.
    
    \begin{theorem} 
        Let $P$ and $Q$ be arbitrary binary-input discrete memoryless channels, and let $\mathcal{R}^*(P,Q)$ be the supremum of learning rates across all teaching and learning protocols. Then
        \begin{equation}
            \mathcal{R}^*(P,Q) \ge -\min_{s \in [0,1]} \max(\mu_{P}(s), \mu_{Q}(s)). \label{dmc-learning-rate}
        \end{equation}
        Moreover, this can be improved to
        \begin{equation}
            \mathcal{R}^*(P,Q) \ge -\min_{s \in [0,1]} \max\big (\mu_{P}(s), \min(\mu_{Q}(s), \mu_{Q}(1-s))\big).
            \label{tight-learning-rate}
        \end{equation}
        \label{dmc-main-result}
    \end{theorem}
    
    While \eqref{tight-learning-rate} clearly implies \eqref{dmc-learning-rate}, we find the expression \eqref{dmc-learning-rate} easier to work with and prove.  Once \eqref{dmc-learning-rate} is proved, \eqref{tight-learning-rate} follows from the fact that one can choose to flip the inputs of $Q$, which amounts to replacing $\mu_Q(s)$ by $\mu_Q(1-s)$.\footnote{At first this leads to the seemingly different learning rate of $-\min\big( \min_{s \in [0,1]}  \max(\mu_{P}(s), \mu_{Q}(s)), \min_{s \in [0,1]}  \max(\mu_{P}(s), \mu_{Q}(1-s)) \big)$, but then \eqref{tight-learning-rate} can be recovered by taking the two $\min_{s \in [0,1]}$ outside the outer minimum (swapping minimization order has no effect), and applying the identity $\max(a,\min(b,c)) = \min( \max(a,b), \max(a,c))$ (verified by checking all 6 orderings of $a$, $b$, and $c$). }
    
    When the channels $P$ and $Q$ are identical, \eqref{dmc-learning-rate} reduces to the expression in Corollary \ref{1-hop-error}. Since the learning rate cannot be smaller than the 1-hop case (due to a data-processing argument similar to the binary symmetric setting), we conclude that Theorem \ref{dmc-main-result} provides the optimal learning rate in this case.  Other scenarios in which Theorem \ref{dmc-main-result} gives the optimal learning rate will be discussed in Section \ref{sec:dmc_conv}.
    
    \subsection{Protocol and Analysis for the Achievability Result} \label{sec:dmc_protocol}
    
    Fix an arbitrary value of $\bar{s} \in [0,1]$ and define
    \begin{equation}
        \mu_{\max} = \max(\mu_{P}(\bar{s}), \mu_{Q}(\bar{s})).
    \end{equation}
    To prove \eqref{dmc-learning-rate}, it suffices to show that $-\mu_{\max}$ is an achievable learning rate.  We assume without loss of generality that $\mu_{\max} < 0$, since a learning rate of zero is trivial.
    
    We again adopt the block structure used in Section \ref{sec:bsc}, with $k$ denoting the block size.  Suppose that in a single block, the teacher receives $\vec{Y} = (Y_1, Y_2,\ldots, Y_k)$.  We define the log-likelihood ratio (LLR) as
    \begin{equation}
        l(\vec{Y}) = \sum_{i=1}^k \ln \frac{P(Y_i|1)}{P(Y_i|0)}, \label{eq:def_l}
    \end{equation}
    and let $L_1$ (respectively, $L_{0}$) be the distribution of $l(\vec{Y})$ under $\Theta=1$ (respectively, $\Theta=0$).
    
    Upon receiving $\vec{Y}$, the teacher computes $l = l(\vec{Y})$, and sends $g(l)$ bits of $0$ followed by $k-g(l)$ bits of $1$, where $g(l)$ is defined as follows:
    \begin{equation}
        g(l) = 
        \begin{cases}
            \min(k, \frac{1-\bar{s}}{\mu_{\max}} \ln \p(L_{0} \geq l)) & l\leq 0\\
            \max(0, k - \frac{\bar{s}}{\mu_{\max}}\ln \p(L_{1} \leq l)) & l> 0
        \end{cases} \label{eq:def_g}
    \end{equation}
    Since $1-\bar{s}\geq 0$, $\mu_{\max}<0$, and $\ln \p(L_{0}\geq l) \leq 0$, we have
    \begin{equation}
        \frac{1-\bar{s}}{\mu_{\max}} \ln \p(L_{0} \geq l) \geq 0, \label{eq:pos1}
    \end{equation}
    and similarly,
    \begin{equation}
        \frac{\bar{s}}{\mu_{\max}}\ln \p(L_{1} \leq l) \geq 0, \label{eq:pos2}
    \end{equation}
    which implies that $0\leq g(l) \leq k$ for all $l$.  It may be necessary to round $g(l)$ to the nearest integer, but similarly to Section \ref{sec:bsc}, this does not affect the result, and it is omitted in our analysis.
    
    The main difference in the analysis compared to Section \ref{sec:bsc} is establishing the following technical lemma. 
    
    \begin{lemma} \label{lem:tech_lem}
        Under the preceding setup and definitions, we have for all $l$ that
        \begin{equation}
            \frac{1-\bar{s}}{\mu_{\max}} \ln \p(L_{0} \geq l) \geq	k - \frac{\bar{s}}{\mu_{\max}} \ln \p(L_1 \leq l).
            \label{f-bound-1}
        \end{equation}
    \end{lemma}
    \begin{proof}
        The proof is based on probabilistic bounds on log-likelihood ratios given in \cite{shannon1967}, along with algebraic manipulations that are elementary but rather technical; the details are given in Appendix \ref{app:tech_lem}.
    \end{proof}
    
    We claim that this result implies
    \begin{equation}
        \frac{1-\bar{s}}{\mu_{\max}}\ln \p(L_{0} \geq l) \geq g(l) \geq k - \frac{\bar{s}}{\mu_{\max}}\ln \p(L_{1} \leq l).
        \label{f-bound-2}
    \end{equation}
    To see this, we check the possible cases in the definition of $g(l)$ in \eqref{eq:def_g}.  If $l \le 0$, then the left inequality in \eqref{f-bound-2} is trivial, and the right inequality follows by considering two sub-cases: If the $\min(k,\cdot)$ in \eqref{eq:def_g} is attained by $k$, then we apply \eqref{eq:pos2}, and otherwise, we apply \eqref{f-bound-1}.  The case $l > 0$ is handled similarly.
    
    From here, the analysis is similar to that of Section \ref{sec:bsc}. Observe that we have the Markov chain $\Theta \rightarrow l(\vec{Y}) \rightarrow W$, where $W$ denotes the $k$ symbols received by the student in a single block.  Since $W$ takes on two different distributions depending on whether $\Theta=1$ or $\Theta=0$, we may also view $W$ as the output of a `single-use' discrete memoryless channel with input $\Theta$ and an output of length $k$.  With a slight abuse of notation, we define the quantities $\mu_W(s)$ and $\rho_W(s)$ according to this channel.
    
    To upper bound $\rho_W(s)$, we use  the following straightforward generalization of Lemma \ref{lem:mixture}.
    
    \begin{lemma}
        Suppose that $A$ follows a mixture distribution described by $[p_1,A_1; \ldots; p_k,A_k]$. Then for any distribution $B$, we have
        \begin{equation}
            \rho(A,B,s) \leq \sum_{i=1}^k p_i^{1-s} \rho(A_i, B, s),
        \end{equation}
        and
        \begin{equation}
            \rho(B,A,s) \leq \sum_{i=1}^k p_i^{s} \rho(B, A_i, s).
        \end{equation}
        \label{general-mixture}
    \end{lemma}
    The proof of this lemma is identical to that of Lemma \ref{lem:mixture}, except that $\sqrt{a+b}\leq \sqrt{a}+\sqrt{b}$ is replaced by the more general $(a+b)^{s} \leq a^{s} + b^{s}$ for $s \in [0,1]$ (and similarly for $1-s$).
    
    Decomposing each distribution $P_{W|\Theta=1}$ and $P_{W|\Theta=0}$ as a mixture based on the (finitely many) possible values of $l$, and applying Lemma \ref{general-mixture} twice, we obtain
    \begin{align}
        \rho_W(\bar{s}) 
        &\leq \sum_{l_{0}, l_1} \p(L_{0} = l_{0})^{1-\bar{s}} \p(L_1 = l_1)^{\bar{s}} \rho(P_{W|l_0}, P_{W|l_1}, \bar{s}) \\
        &\leq \sum_{l_{0}, l_1} \p(L_{0} \ge l_{0})^{1-\bar{s}} \p(L_1 \ge l_1)^{\bar{s}} \rho(P_{W|l_0}. P_{W|l_1}, \bar{s}). \label{eq:rhoW}
    \end{align}
    To help simplify this expression, we use the following consequence of \eqref{f-bound-2} (again recalling $\mu_{\max} < 0$):
    \begin{multline}
        \p(L_{0}\geq l_{0})^{1-\bar{s}} \p(L_1\geq l_1)^{\bar{s}} \\ \leq  \exp(g(l_0) \cdot \mu_{\max})  \exp((k-g(l_1))\cdot \mu_{\max}).  \label{eq:useful}
    \end{multline}	
    Suppose first that $l_0 \leq l_1$.  Since $g$ is an increasing function and only the bits in between positions $g(l_0)$ and $g(l_1)$ differ, we have
    \begin{align}
        \rho(P_{W|l_0}, P_{W|l_1}, \bar{s}) 
        &= \exp(\mu_Q(\bar{s}) \cdot (g(l_1)-g(l_0)))  \\
        &\leq  \exp(\mu_{\max}\cdot (g(l_1)-g(l_0))), \label{eq:case1}
    \end{align}
    noting that each differing bit contributes $e^{\mu_Q(\bar{s})} \le e^{\mu_{\max}}$ to the total.  Combining \eqref{eq:useful} and \eqref{eq:case1}, the terms containing $g(\cdot)$ cancel, and we are left with
    \begin{multline}
        \sum_{l_0 \leq l_1} \p(L_{0}\geq l_{0})^{1-\bar{s}} \p(L_1\geq l_1)^{\bar{s}} \rho(P_{W|l_0}, P_{W|l_1}, \bar{s}) \\ \leq \sum_{l_0 \leq l_1} \exp(k \cdot \mu_{\max}). \label{eq:final_bound0}
    \end{multline}
    On the other hand, if $l_0 > l_1$, then the increasing property of $g$ gives
    \begin{align}
        &\p(L_{0}\geq l_{0})^{1-\bar{s}} \p(L_1\geq l_1)^{\bar{s}}  \nonumber \\
        &\quad\leq \exp(g(l_0) \cdot \mu_{\max}) \exp((k-g(l_1)) \mu_{\max}) \\
        &\quad\leq \exp(k \cdot \mu_{\max}),
    \end{align}
    and combining this with the fact that $\rho \leq 1$, we obtain
    \begin{multline}
        \sum_{l_{0} > l_1} \p(L_{0}\geq l_{0})^{1-\bar{s}} \p(L_1\geq l_1)^{\bar{s}} \rho(P_{W|l_0}, P_{W|l_1}, \bar{s}) \\ \leq \sum_{l_{0} > l_1} \exp(k \cdot \mu_{\max}). \label{eq:final_bound}
    \end{multline}
    
    
    By the definition of $l(\vec{Y})$ in \eqref{eq:def_l}, the value of $l$ only depends on the number of occurrences of each output symbol in the block of length $k$, and not the specific order.  Hence, with a finite output alphabet $\mathcal{Y}$, the number of possible $l$ values is upper bounded by $(k+1)^{|\mathcal{Y}|}$ \cite[Ch.~2]{Csi11}, and we conclude from \eqref{eq:rhoW}, \eqref{eq:final_bound0} and \eqref{eq:final_bound} that
    \begin{equation}
        \mu_W(\bar{s}) = \ln \rho_W(\bar{s}) \leq k \cdot \mu_{\max} + \mathcal{O}(\log k).
    \end{equation}
    As with \eqref{rate-2}, this implies that the overall learning rate $\mathcal{R}$ satisfies
    \begin{equation}
        \mathcal{R} \geq -\frac{1}{k} \mu_W(\bar{s}) = -\mu_{\max} + \mathcal{O}\left(\frac{\log k}{k}\right).
    \end{equation}
    By setting $k$ large enough, we can obtain any rate arbitrarily close to $-\mu_{\max}$, completing the proof of Theorem \ref{dmc-main-result}.
    
    \subsection{A Converse Result and Its Consequences} \label{sec:dmc_conv}
    
    As hinted following Theorem \ref{dmc-main-result}, by data processing arguments, we cannot obtain a learning rate greater than
    \begin{equation}
        -\max\Big(\min_{s \in [0,1]} \mu_P(s), \min_{s \in [0,1]} \mu_Q(s)\Big),
        \label{naive-bound}
    \end{equation}
    and this matches Theorem \ref{dmc-main-result} when $P=Q$.  More generally, we can immediately deduce that Theorem \ref{dmc-main-result} gives the optimal learning rate whenever $Q$ is a {\em degraded} version of $P$ (i.e., there exists an auxiliary channel $U$ such that composing $P$ with $U$ gives an overall channel with the same transition law as $Q$), or vice versa.  This is because the teacher (respectively, student) can artificially introduce extra noise to reduce to the case that $P=Q$.
    
    In this section, we give an additional upper bound on the learning rate (i.e., converse) that improves on the simple one in \eqref{naive-bound}, and establishes the tightness of Theorem \ref{dmc-main-result} in certain cases beyond those mentioned above.
    To state this converse, we introduce an additional definition.  Momentarily departing from the teaching/learning setup, consider the following two-round 1-hop communication scenario defined in terms of the DMCs $P,Q$ and a parameter $\gamma \in [0,1]$:
    \begin{itemize}
        \item The sender seeks to communicate one of two messages, $\Theta \in \{0,1\}$, each occurring with probability $\frac{1}{2}$;
        \item In the first round of communication, the sender transmits $\gamma n$ symbols via independent uses of $P$ without feedback;
        \item After this first round, the sender observes the corresponding $\gamma n$ output symbols via a noiseless feedback link;
        \item In the second round of communication, the sender transmits $(1-\gamma)n$ symbol via independent uses of $Q$, with no further feedback.
    \end{itemize}
    This is a somewhat unconventional setup with independent and non-identical channel uses, as well as partial feedback.  We let $E_{2}(P,Q,\gamma)$ denote the best possible error exponent under this setup as $n \to \infty$, with the subscript 2 representing its two-round nature.
    
    \begin{theorem} \label{thm:conv}
        Let $P$ and $Q$ be arbitrary binary-input discrete memoryless channels, and let $\mathcal{R}^*(P,Q)$ be the supremum of learning rates across all teaching and learning protocols. Then, for any  $\gamma \in [0,1]$, we have
        \begin{equation}
            \mathcal{R}^*(P,Q) \le E_2(P,Q,\gamma).
        \end{equation}
    \end{theorem}
    \begin{proof}
        Consider a ``genie-aided'' setup in which the teacher and student are given the following additional information:
        \begin{itemize}
            \item[(i)] The student can observe the input $\hat{X}$ (rather than only the output $Z$) of the teacher-student channel $Q$ from time $1$ up to $\gamma n$.
            \item[(ii)] The teacher is given the true value of $\Theta$ after time $\gamma n$.
        \end{itemize}
        Since the teacher and student can always choose to ignore this additional information, any upper bound on the learning rate in this genie-aided setting is also an upper bound in the original setting.
        
        The theorem follows by noting that this genie-aided setup can be viewed as an instance of the above two-round communication setup:
        \begin{itemize}
            \item For the first $\gamma n$ symbols, once the student is given $\hat{X}$, no further information is provided by $Z$ (a ``degraded'' version of $\hat{X}$), and the student's first $\gamma n$ observed symbols correspond to repeatedly observing $\Theta$ through the channel $P$ without feedback.
            \item Similarly, for the last $(1-\gamma)n$ symbols, once the teacher is given $\Theta$, no further information is provided by $Y$, but the value of $\Theta$ still needs to be conveyed to the student via the uses of $Q$.  Hence, the student's final $(1-\gamma) n$ observed symbols correspond to using $Q$ with no further feedback, but with the first $\gamma n$ outputs being available at the sender (teacher).
        \end{itemize}
    \end{proof}
    
    In the remainder of this subsection, we connect the converse result of Theorem \ref{thm:conv} with the achievability result of Theorem \ref{dmc-main-result} in two different ways:
    \begin{itemize}
        \item Under a certain technical assumption (Assumption \ref{assump1} below), we show that there exists $\gamma \in [0,1]$ such that the exponent in Theorem \ref{dmc-main-result} equals $E_1(P,Q,\gamma)$, which is defined in the same way as $E_2(P,Q,\gamma)$ but with no feedback being available in the communication setup.  Hence, the achievability and converse only differ with respect to the availability of feedback.
        \item In the case that $Q$ is a BSC, we show that Theorems \ref{dmc-main-result} and \ref{thm:conv} coincide for all $P$, thus establishing the optimal learning rate.
    \end{itemize}
    Starting with the former, we consider the following technical assumption.
    
    \begin{figure}
        \begin{centering}
            \includegraphics[width=0.95\columnwidth]{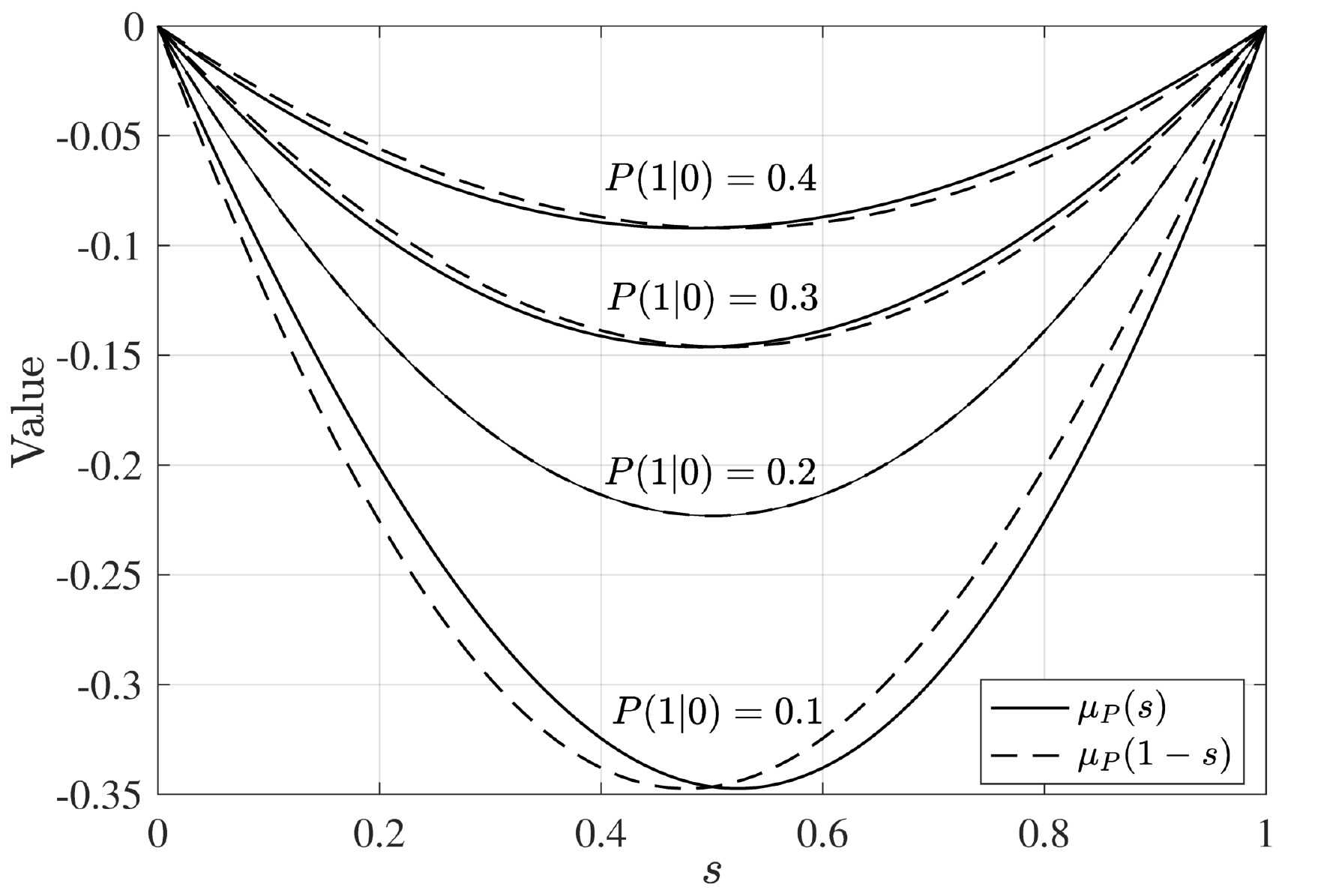}
            \par
        \end{centering}
        
        \caption{Example comparisons of $\mu_P(s)$ vs.~$\mu_P(1-s)$ for binary-input binary-output channels.  We fix $P(0|1) = 0.2$, and vary $P(1|0)$.  We observe that for every curve shown, it either holds that $\mu_P(s) \le \mu_P(1-s)$ for all $s \in [0,\frac{1}{2}]$, or $\mu_P(s) \ge \mu_P(1-s)$ for all $s \in [0,\frac{1}{2}]$. \label{fig:mu_examples}}
    \end{figure}
    
    \begin{assump} \label{assump1}
        The DMCs $P$ and $Q$ satisfy the property that for all $s \in \big[0,\frac{1}{2}\big]$, it holds that $\mu_P(s) \le \mu_P(1-s)$ and  $\mu_Q(s) \le \mu_Q(1-s)$.
    \end{assump}
    
    This assumption states that values of $s \in \big[0,\frac{1}{2}\big]$ are uniformly ``better'' ($\mu$ is more negative) than their counterparts flipped about $\frac{1}{2}$.  The opposite scenario (i.e., $\mu(s) \ge \mu(1-s)$ for all $s \in \big[0,\frac{1}{2}\big]$) can be handled similarly, or can more simply be viewed as reducing to Assumption \ref{assump1} upon swapping the two inputs.  Up to such swapping, we were unable to find any binary-input binary-output channel for which Assumption \ref{assump1} fails, and even in the multiple-output case, we found that it is difficult to find counter-examples.  As two concrete examples, the BSC trivially satisfies $\mu(s) = \mu(1-s)$, and the reverse Z-channel\footnote{This is the regular Z-channel with the inputs swapped.} with parameter $q \in (0,1)$ yields the increasing function $\mu(s) = (1-s)\ln q$.  See Figure \ref{fig:mu_examples} for further illustrative numerical examples with binary output.

    \begin{lemma} \label{lem:non_fb_exponent}
        For any binary-input DMCs $P$ and $Q$ satisfying Assumption \ref{assump1}, we have 
        \begin{equation}
            \min_{\gamma \in [0,1]} E_1(P,Q,\gamma) = -\min_{s \in [0,1]} \max\big (\mu_{P}(s), \mu_{Q}(s))\big), \label{eq:E1}
        \end{equation}
        where $E_1(P,Q,\gamma)$ is defined in the same way as $E_2(P,Q,\gamma)$ above, but with no feedback being available in the communication setup. 
    \end{lemma}
    \begin{proof}
        See Appendix \ref{app:non_fb_exponent}
    \end{proof}
    
    Observe that \eqref{eq:E1} precisely equals the learning rate derived in Theorem \ref{dmc-main-result}, and that we would match the converse in Theorem \ref{thm:conv} if we could additionally show that $E_1 = E_2$.  Unfortunately, even the single round of feedback can strictly increase the error exponent (see Appendix \ref{app:fb_example} for an example), meaning that the optimal learning rate remains unclear.  On the positive side, it turns out that $E_1 = E_2$ whenever $Q$ is a BSC, and in fact, in this case, we can establish that Theorem \ref{dmc-main-result} gives the optimal learning rate even when Assumption \ref{assump1} is dropped.  Formally, we have the following.
    
    \begin{lemma} \label{lem:symmQ}
        For any binary-input DMC $P$, if $Q$ is a BSC with crossover probability $q \in \big(0,\frac{1}{2}\big)$, then the achievable learning rate in \eqref{dmc-learning-rate} is tight, i.e., $\mathcal{R}^*(P,Q) = -\min_{s \in [0,1]} \max(\mu_{P}(s), \mu_{Q}(s))$.
    \end{lemma}
    \begin{proof}
        See Appendix \ref{app:symmQ}.
    \end{proof}
    
    Using Lemma \ref{lem:symmQ} as a starting point, it is straightforward to establish that Theorem \ref{thm:conv} can indeed strictly improve on the trivial converse in \eqref{naive-bound}.
    
    \begin{cor} \label{cor:strict_improvement}
        There exist binary-input DMCs $P$ and $Q$ such that the optimal learning rate is strictly smaller than $-\max\big(\min_{s \in [0,1]} \mu_P(s), \min_{s \in [0,1]} \mu_Q(s)\big)$.
    \end{cor}
    \begin{proof}
        This follows by letting $P$ be a (reverse) Z-channel and $Q$ be a BSC, with suitably-chosen parameters.
        See Appendix \ref{app:strict_improvement} for the details.
    \end{proof}
    
    \subsection{A Converse for Block-Structured Protocols} \label{sec:block_conv}
    
    While we have established several special cases where Theorem \ref{dmc-main-result} give the optimal learning rate, it remains open as to whether this is true in general.  As a final result, we complement our general converse (Theorem \ref{thm:conv}) with a {\em restricted} converse that only applies to block-structured protocols.  Specifically, we say that a teaching strategy is {\em block structured with length $k$} if it follows the structure in Figure \ref{block-protocol}:  For any positive integer $i$, the bits transmitted at indices $ki+1, \dotsc, k(i+1)$ are only allowed to depend on the bits received at indices $k(i-1)+1,\dotsc,ki$.
    
    \begin{theorem}	\label{thm:block_conv}
        For any $P$ and $Q$ satisfying Assumption \ref{assump1}, and any protocol such that the teaching strategy is block structured with length $k = o(n)$, the learning rate is at most $-\min_{s \in [0,1]} \max(\mu_{P}(s), \mu_{Q}(s))$. 
    \end{theorem}
    \begin{proof}
        See Appendix \ref{app:block_converse}.
    \end{proof}
    
    This result serves as evidence that Theorem \ref{dmc-main-result} may provide the optimal learning rate under Assumption \ref{assump1}, or at least that to have any hope of improving, the teaching strategy should have some form of {\em long-term memory}.
    
    %
    
    \section{Conclusion and Discussion} \label{sec:conclusion}
    
    We have established optimal learning rates for the problem of teaching and learning a single bit under binary symmetric noise, using a simple and computationally efficient block-structured strategy. We have also adapted this technique to general binary-input DMCs.
    
    As discussed above, while the optimal learning rate for general binary-input DMCs follows in several cases of interest, it remains unknown in general, leaving us with the following open problem.
    
    \begin{conj}
        Does \eqref{tight-learning-rate} give the optimal learning rate for arbitrary binary-input DMCs $P$ and $Q$?
    \end{conj}

    An important step towards answering this question would be removing Assumption \ref{assump1} in the results that currently make use of it.  Also regarding this assumption, we expect that it should hold true for all binary-input channels, and it would be of interest to establish whether this is indeed the case.
    
    Another possible direction is to increase the input alphabet size of $Q$, removing the assumption of binary inputs. In the 1-hop case, having a channel with more inputs does not improve the learning rate beyond choosing the two best inputs for the channel \cite{shannon1967}.  By a similar idea, we can immediately deduce an achievability result via Theorem \ref{dmc-main-result}, but we are left with the following open problem.
    
    \begin{conj}
        Does there exist a pair of DMCs $P, Q$, such that every restriction of $Q$ to two inputs gives a strictly suboptimal learning rate?
    \end{conj}
    
    In addition to  these open problems, interesting directions may include the case that the channel transition laws are unknown to the teacher and/or student, and the case of teaching more than a single bit.
    
    \appendix
    
    \subsection{Proof of Lemma \ref{lem:tech_lem} (Technical Result for General Binary-Input DMCs)} \label{app:tech_lem}
    
    We first state an auxiliary result from \cite{shannon1967}.  When the agent seeks to decide between the alternatives $\Theta=1$ and $\Theta=0$, the natural decision rule is to set a threshold on the corresponding log-likelihood ratio. The following result bounds the probability that this log-likelihood ratio test fails under a given threshold.
    
    \begin{lemma}
        {\em \cite[Proof of Thm.~5]{shannon1967}} Let $P$ be a binary-input DMC, and for each output symbol $y$, define
        \begin{equation}
            l(y) = \ln \frac{P(y|1)}{P(y|0)}
        \end{equation}
        to be the log-likelihood ratio between the two conditional distributions.  Let $L_1$ and $L_{0}$ follow the distribution of $l(Y)$ under $Y \sim P(\cdot | 1)$ and $Y \sim P(\cdot |0)$ respectively.   For all $s \in (0,1)$, we have
        \begin{equation}
            \p(L_{0} \geq \mu_P'(s))
            \leq \exp[\mu_P(s) - s\mu_P'(s)],
            \label{bound1}
        \end{equation}
        and
        \begin{equation}
            \p(L_1 \leq \mu_P'(s)) \leq \exp[\mu_P(s) + (1-s)\mu_P'(s)].
            \label{bound2}
        \end{equation}
        \label{thm:error-bound}
    \end{lemma}
    
    We now proceed with the proof of Lemma \ref{lem:tech_lem}.  First suppose that $l$ satisfies
    \begin{equation}
        \lim_{s\rightarrow 0^+} \mu'_P(s) \leq l \leq \lim_{s\rightarrow 1^-} \mu'_P(s).
        \label{l-bound}
    \end{equation}
    Then, since $\mu'_P(s)$ is continuous and non-decreasing \cite{shannon1967}, there exists $t \in [0,1]$ such that $l = \mu_P'(t)$, and by \eqref{bound1} and \eqref{bound2} (applied to the $k$-fold product distribution), we obtain
    \begin{equation}
        \ln \p(L_{0} \geq l) \leq k(\mu_P(t) - t \mu_P'(t)),
        \label{ll1}
    \end{equation}
    and 
    \begin{equation}
        \ln \p(L_1 \leq l) \leq k(\mu_P(t) + (1-t) \mu_P'(t)).
        \label{ll2}
    \end{equation}
    On the other hand, if
    \begin{equation}
        l < \lim_{s\rightarrow 0^+} \mu'_P(s),
    \end{equation}
    then we may set $t=0$, and \eqref{ll1} still holds due to the fact that $\mu_P(0)=0$. In this case, we have
    \begin{equation}
        \ln \p(L_1\leq l) \leq \ln \p\Big(L_1 \leq \lim_{s\rightarrow 0^+} \mu_P'(s)\Big) \stackrel{(a)}{\leq} k \mu_P'(0),
    \end{equation}
    where (a) follows from fact that \eqref{ll2} holds when $l = \mu_P'(t)$.  It follows that \eqref{ll2} also holds here with $t=0$.  By a similar argument, if $l > \lim_{s\rightarrow 1^-} \mu'_P(s)$, then we may set $t=1$, and \eqref{ll1} and \eqref{ll2} remain valid.  Therefore, for every value of $l$ (possibly $\pm \infty$), there exists some $t \in [0,1]$ such that equations \eqref{ll1} and \eqref{ll2} hold.
    
    
    Since $\mu_P$ is convex (see \cite[Thm.~5]{shannon1967}), $\mu_P(\bar{s})$ is lower bounded by its tangent line approximation at $t$.   This implies
    \begin{equation}
        \mu_P(t) + (\bar{s}-t) \mu_P'(t) \leq \mu_P(\bar{s}) \leq \mu_{\max},
    \end{equation}
    which re-arranges to give
    \begin{equation}
        (1-\bar{s})(\mu_P(t)-t\mu_P'(t)) \leq \mu_{\max} - \bar{s}(\mu_P(t)+(1-t)\mu_P'(t)).
    \end{equation}
    Multiplying  through by $k > 0$ and dividing through by $\mu_{\max} < 0$, we obtain
    \begin{equation}
        \frac{1-\bar{s}}{\mu_{\max}} k(\mu_P(t) - t \mu_P'(t)) \geq	k - \frac{\bar{s}}{\mu_{\max}} k(\mu_P(t) + (1-t) \mu_P'(t)).
        \label{tangent-line}
    \end{equation}
    Finally, further upper bounding the left-hand side via \eqref{ll1} and further lower bounding the right-hand side via \eqref{ll2}, we deduce the desired inequality in \eqref{f-bound-1}.
    
    \subsection{Proof of Lemma \ref{lem:non_fb_exponent} (1-Hop Exponent Without Feedback)} \label{app:non_fb_exponent}
    
    Let $\vec{x}$ and $\vec{x}'$ be the two codewords in the 1-hop communication problem without feedback, let $\vec{y}$ denote the received sequence, and define the $n$-letter version of $\mu$ as
    \begin{equation}
        \mu_{n,\gamma}(s) = \ln \sum_{\vec{y}} \p( \vec{Y} = \vec{y} | \vec{X} = \vec{x}' )^{1-s} \p( \vec{Y} = \vec{y} \,|\, \vec{X}= \vec{x} )^{s}.
    \end{equation}
    Using the classical analysis of Shannon, Gallager, and Berlekamp \cite[Thm.~5]{shannon1967} (see also Lemma \ref{lem:berlekamp} regarding the case $P=Q$), the error exponent attained by optimal (maximum-likelihood) decoding is \cite[Thm.~5]{shannon1967}\footnote{The $\sqrt{\mu''_n(s)}$ term in \cite[Thm.~5]{shannon1967} scales as $O(\sqrt n)$ due to the additive property of $\mu$ for product distributions, and thus, it does not affect the error exponent.}
    \begin{equation}
        \lim_{n \to \infty}  -\frac{1}{n} \min_{s \in [0,1]} \mu_n(s), \label{eq:lim_mu}
    \end{equation}
    and accordingly, we proceed by analyzing $\mu_n(s)$.
    
    Without loss of optimality, we can assume that $\vec{x}$ and $\vec{x}'$ differ in every entry.  Let $\beta_P$ be the fraction of uses of $P$ for which $\vec{x}$ takes value $1$ (and hence $\vec{x}'$ takes value $0$), and define $\beta_Q$ analogously.  Then, the additive property of $\mu$ for product distributions (by taking the logarithm in \eqref{mult2}) gives
    \begin{multline}
        \mu_{n,\gamma}(s) = \beta_P\gamma n \mu_P(s) + (1-\beta_P)\gamma n \mu_P(1-s) \\ + \beta_Q(1-\gamma) n \mu_Q(s) + (1-\beta_Q) (1-\gamma) n \mu_Q(1-s),
        \label{86}
    \end{multline}
    noting that by the definition of $\mu_P$, swapping the two inputs amounts to replacing $s$ by $1-s$.
    
    To further simplify \eqref{86}, we define
    \begin{gather}
        \mu^* = \min_{s \in [0,1]} \, \max(\mu_P(s), \mu_Q(s)),
        \label{87} \\
        s^* = \argmin_{s \in [0,1]} \, \max(\mu_P(s), \mu_Q(s)),
        \label{88}
    \end{gather}
    and introduce the shorthand
    \begin{equation}
        \mutilde_{n,\gamma}(s) = \gamma n \mu_P(s) + (1-\gamma)n \mu_Q(s). \label{eq:mutilde}
    \end{equation}
    We first provide a lemma lower bounding $\mutilde_{n,\gamma}(s)$, and then show that $\mutilde_{n,\gamma}(s)$ is itself a lower bound on $\mu_{n,\gamma}(s)$.
    
    \begin{lemma} \label{lem:gamma} There exists a choice of $\gamma \in [0,1]$ such that for all $s \in [0,1]$, it holds that $\mutilde_{n,\gamma}(s) \geq n \mu^*$.
    \end{lemma}
    \begin{proof}
        We split the proof into several cases.  Throughout, it should be kept in mind that $\mu_P(s)$ is a continuous and convex function \cite{shannon1967}.
        
        \underline{Case 1}: If $\mu_P(s^*) \geq \mu_Q(s^*)$ and $\mu_P$ attains its global minimum\footnote{Here ``global minimum'' is meant with respect to the restricted domain $s \in [0,1]$.} at $s^*$, then we can set $\gamma=1$ to get
        \begin{align}
            \mutilde_{n,\gamma}(s) &= n\mu_P(s) \geq n\mu_P(s^*) \nonumber \\ &= n\cdot \max(\mu_P(s^*), \mu_Q(s^*)) = n\mu^*.
        \end{align}
        A similar argument holds when $\mu_Q(s^*)\geq \mu_P(s^*)$ and $\mu_Q$ attains its global minimum at $s^*$ (setting $\gamma=0$).
        
        \underline{Case 2}: If $\mu_P(s^*) > \mu_Q(s^*)$ and $\mu_P$ does not attain its global minimum at $s^*$, then we can shift $s^*$ towards $P$'s minimizer to decrease $\max(\mu_P(s), \mu_Q(s))$, contradicting the definition of $s^*$.  Hence, this case does not occur.  A similar finding applies to the case that $\mu_Q(s^*) > \mu_P(s^*)$ and $\mu_Q$ does not attain its global minimum at $s^*$
        
        \underline{Case 3}: We claim that the above cases only leave the case that $\mu_P(s^*) = \mu_Q(s^*)$ and neither $\mu_P$ nor $\mu_Q$ attain their global minimum at $s^*$.  To see this, note that if $\mu_P(s^*) \neq \mu_Q(s^*)$, then the larger one attaining its global minimum would give Case 1, whereas the larger one not attaining its global minimum would give Case 2.  Hence, we may assume that $\mu_P(s^*) = \mu_Q(s^*)$.  Then, we can additionally assume that neither $\mu_P$ nor $\mu_Q$ attain their global minimum at $s^*$, since otherwise we would again be in Case 1.  We proceed with two further sub-cases.
        
        {\em Case 3a:} If $\mu'_P(s^*)\mu'_Q(s^*)>0$, then we can shift $s^*$ by a small amount and decrease $\max(\mu_P(s), \mu_Q(s))$, contradicting the fact that $s^*$ is the minimizer (unless we are at an endpoint of the interval $[0,1]$, in which case one of $\mu_P$ and $\mu_Q$ must be a global minimum, again a contradiction).  Hence, this sub-case does not occur.
        
        {\em Case 3b:} If $\mu'_P(s^*)\mu'_Q(s^*)\leq 0$, then we choose
        \begin{equation}
            \gamma = \frac{\mu'_Q(s^*)}{\mu'_Q(s^*)-\mu'_P(s^*)}, \quad 1-\gamma = \frac{\mu'_P(s^*)}{\mu'_P(s^*)-\mu'_Q(s^*)}.
        \end{equation}
        The assumption $\mu'_P(s^*)\mu'_Q(s^*)\leq 0$ ensures that these quantities both lie in $[0,1]$.  Then, taking the derivative in \eqref{eq:mutilde}, we obtain
        \begin{equation}
            \mutilde'_{n,\gamma}(s) = (\gamma n \mu'_P(s) + (1-\gamma) n \mu'_Q(s))|_{s=s^*} = 0.
        \end{equation}
        Since $\mutilde_{n,\gamma}$ is convex (following from $\mu_P$ and $\mu_Q$ being convex), this implies that $\mutilde_{n,\gamma}$ must attain its global minimum at $s=s^*$. In addition, since $\mu_P(s^*) = \mu_Q(s^*)$, their common value must equal $\mu^*$, which further implies $\mutilde_{n,\gamma}(s) \geq n \cdot \mu^*$. 
    \end{proof}
    
    We now relate $\mu_{n,\gamma}(s)$ to $\mutilde_{n,\gamma}(s)$.  Recall from Assumption \ref{assump1} that $\mu_P(s) \le \mu_P(1-s)$ and  $\mu_Q(s) \le \mu_Q(1-s)$.  If $s \in [0,\frac{1}{2}]$, then substituting these into \eqref{86} gives
    \begin{equation}
        \mu_{n,\gamma}(s) \geq \mutilde_{n,\gamma}(s). \label{eq:final1a}
    \end{equation}
    On the other hand, if $s \in [\frac{1}{2},1]$, then Assumption \ref{assump1} gives $\mu_{n,\gamma}(s) \geq \mutilde_{n,\gamma}(1-s)$, and substitution into \eqref{86} gives 
    \begin{equation}
        \mu_{n,\gamma}(s) \geq \mutilde_{n,\gamma}(1-s). \label{eq:final1b}
    \end{equation}
    Combining these two cases, we have for all $s \in [0,1]$ that $\mu_{n,\gamma}(s) \geq \mutilde_{n,\gamma}(s')$ for one of the two choices $s' \in \{s,1-s\}$, and hence, Lemma \ref{lem:gamma} implies there exists $\gamma \in [0,1]$ such that
    \begin{equation}
        \min_{s \in [0,1]} \mu_{n,\gamma}(s) \ge n\mu^*. \label{eq:final2}
    \end{equation}
    In view of the fact that \eqref{eq:lim_mu} is the optimal error exponent, we have shown that there always exists $\gamma \in [0,1]$ such that the error exponent is upper bounded by $-\mu^*$. 
    
    To complete the proof, we also need to show that for all $\gamma \in [0,1]$, there exist choices of the codewords and $s \in [0,1]$ such that $\mu_{n,\gamma}(s) \le n\mu^*$.  This is much simpler, following directly by setting $\beta_P = \beta_Q = 1$ (i.e., choosing the all-zeros and all-ones codewords) in \eqref{86}, setting $s = s^*$, upper bounding both $\mu_P(s^*)$ and $\mu_Q(s^*)$ by $\mu^*$.
    
    \subsection{A Case Where Feedback Increases the 1-Hop Exponent} \label{app:fb_example}
    
    Let $P$ be a BSC with crossover probability $p \in \big(0,\frac{1}{2}\big)$, and let $Q$ be a reverse Z-channel (RZ-channel for short) with $Q(1|1) = 1$ and $Q(1|0) = q \in (0,1)$.  Note that Assumption \ref{assump1} holds for these channels, as discussed just above Lemma \ref{lem:non_fb_exponent}.
    Consider the following communication strategy with partial feedback:
    \begin{itemize}
        \item Repeatedly input $\Theta$ to the BSC in the first $\gamma n$ channel uses.
        \item After observing the feedback, if more than half the BSC bits are flipped, send the all-$0$s sequence over the RZ-Channel in $(1-\gamma)n$ uses, and otherwise send the all-$1$s sequence.
    \end{itemize}
    At the decoder, an initial estimate $\tilde{\Theta}$ is formed based on maximum-likelihood decoding (i.e., a majority vote) on the $\gamma n$ received BSC symbols.  Then, after receiving the RZ-channel output, the decision is reversed (i.e., $\hat{\Theta} = 1-\tilde{\Theta}$) if any $0$s are observed among the $(1-\gamma)n$ received symbols; otherwise, the decision is kept (i.e., $\hat{\Theta} = \tilde{\Theta}$).  Essentially, the Z-channel uses are used to tell the decoder whether the initial decision should be changed or not.
    
    If the initial decision $\tilde{\Theta}$ is correct, then the final decision is correct with conditional probability one, since the all $1$s sequence is received noiselessly.  On the other hand, if the initial decision is incorrect, then the probability that it fails to be reversed is $q^{(1-\gamma)n}$.  Hence, an error only occurs if both steps fail, and the overall error exponent $E$ of this strategy is given by
    \begin{equation}
        E = \gamma E_{\rm BSC} + (1-\gamma) \ln \frac{1}{q}, \label{eq:E_fb}
    \end{equation}
    where $E_{\rm BSC} = D\big(\frac{1}{2} \,\big\| p\big) = \frac{1}{2} \ln \big(\frac{1}{4p(1-p)}\big)$ is the optimal error exponent for transmitting two messages over a BSC (see \eqref{eq:1hop_bsc}).
    
    In Figure \ref{fig:fb_exponent}, we compare the exponent in \eqref{eq:E_fb} to that without feedback (Lemma \ref{lem:non_fb_exponent}), setting $p = 0.2$ and $q = 0.8$, which turns out to make the optimal exponents of $P$ and $Q$ identical.  In this example, we see that feedback strictly increases the exponent for all $\gamma \in (0,1)$.  Intuitively, the improvement comes from being able to ``reverse'' initial ``bad events'' to turn them into good events, which is not possible in the absence of feedback.
    
    \begin{figure}
        \begin{centering}
            \includegraphics[width=0.95\columnwidth]{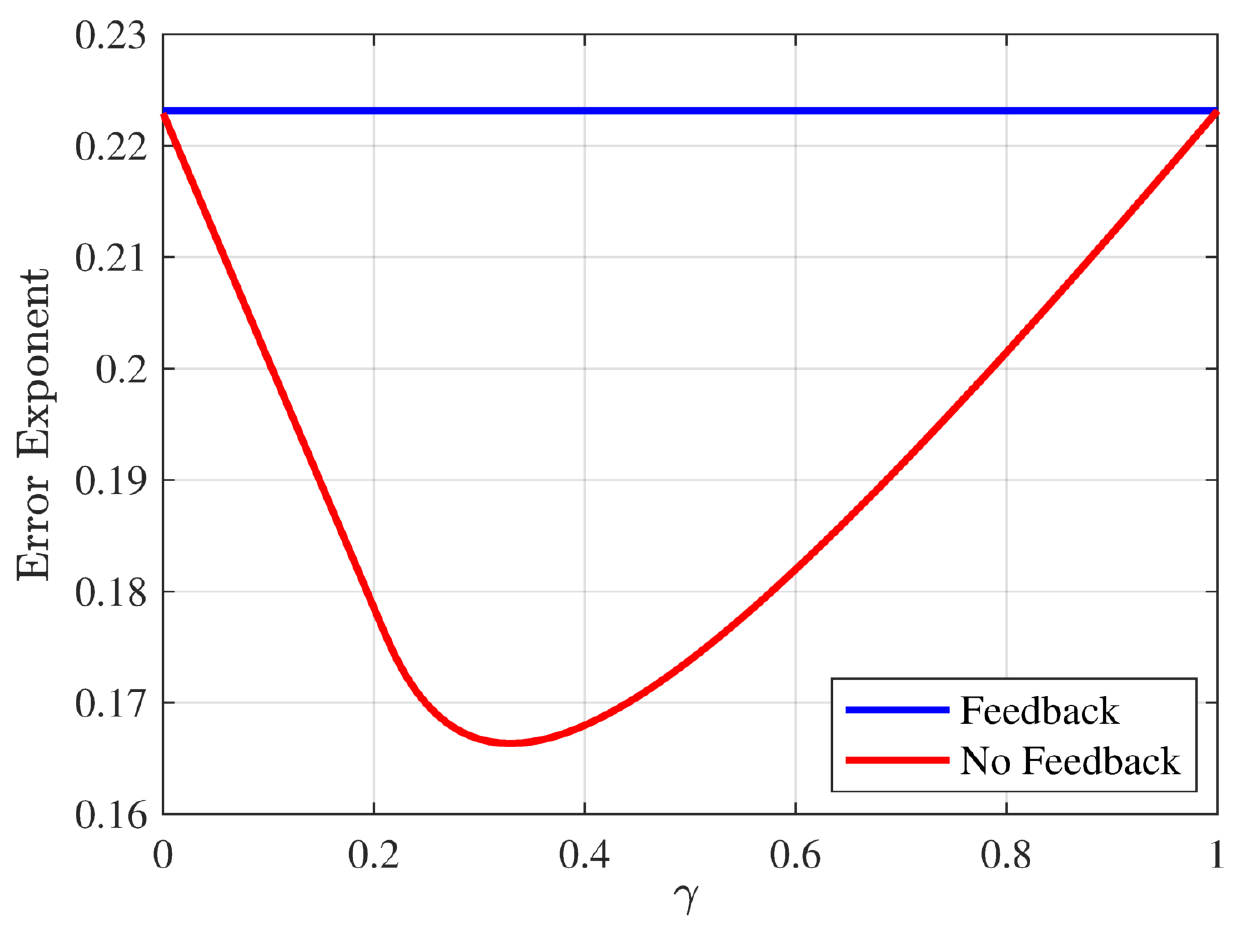}
            \par
        \end{centering}
        
        \caption{Comparison of the error exponent under the two-round coding scheme with feedback vs.~the error exponent without feedback, when $P$ is a BSC with parameter $p = 0.2$, and $Q$ is a Z-channel with parameter $q = 0.8$. \label{fig:fb_exponent}}
    \end{figure}
    
    \subsection{Proof of Lemma \ref{lem:symmQ} (Tightness for Symmetric $Q$)} \label{app:symmQ}
    
    Throughout the proof, we let $\p(\vec{y}|0)$ be a shorthand for $\p(\vec{Y}=\vec{y}|\Theta=0)$, and similarly for $\p(\vec{y}|1)$ and other analogous quantities.
    
    We make use of Theorem \ref{thm:conv}, and accordingly consider a (feedback) communication scenario with $\gamma n$ uses of $P$ and $(1-\gamma)n$ uses of $Q$.  Since we assume $\Theta$ to be equiprobable on $\{0,1\}$, the optimal decoding rule given a received sequence $\vec{y}$ is to set $\hat{\Theta} = 1$ if $\p(\vec{Y}=\vec{y}|\Theta=1) > \p(\vec{Y}=\vec{y}|\Theta=0)$, and $\hat{\Theta} = 0$ otherwise.  Under this optimal rule, the error probability $P_{\rm e} = \p(\hat{\Theta} \ne \Theta)$ is given by \cite[Eq.~(3.2)]{shannon1967}
    \begin{align}
        P_{\rm e} = \frac{1}{2} \sum_{\vec{y}} \min(  \p(\vec{y}|0), \p(\vec{y}|1) ). \label{eq:pe_init}
    \end{align}
    
    We momentarily consider the case that there is no feedback, in which we can simply define $\vec{x}$ and $\vec{x}'$ to be the two length-$n$ codewords (one per message), and choose these codewords to minimize the error probability,  yielding
    \begin{align}
        P_{\rm e}^{\text{no-fb}} = \frac{1}{2} \min_{\vec{x},\vec{x}'} \sum_{\vec{y}} \min( \p(\vec{y}|\vec{x}), \p(\vec{y}|\vec{x}') ). \label{eq:no-fb}
    \end{align}
    Since this expression is still exactly equal to smallest possible error probability, its corresponding exponent is $E_1(P,Q,\gamma)$, defined analogously to $E_2(P,Q,\gamma)$ but without feedback.
    
    In the two-round setting with feedback considered in Theorem \ref{thm:conv}, slightly more care is needed.  We proceed by explicitly splitting up the codeword and received sequence as $\vec{x} = (\vec{x}_P,\vec{x}_Q)$ and $\vec{y} = (\vec{y}_P,\vec{y}_Q)$, and writing the following for $\theta \in \{0,1\}$:
    \begin{equation}
        \p(\vec{y}|\theta) = P^{\gamma n}(\vec{y}_P|\vec{x}_P(\theta))Q^{(1-\gamma)n}(\vec{y}_Q|\vec{x}_Q(\vec{y}_P,\theta)),
    \end{equation}
    where $\vec{x}_Q(\vec{y}_P,\theta)$ explicitly writes $\vec{x}_Q$ as a function of $\vec{y}_P$ and $\theta$, and similarly  for $\vec{x}_P = \vec{x}_P(\theta)$.  Substituting into \eqref{eq:pe_init}, we obtain
    \begin{align}
        &P_{\rm e}=\frac{1}{2} \sum_{\vec{y}_P,\vec{y}_Q} \min\big(  P^{\gamma n}(\vec{y}_P|\vec{x}_P(0))Q^{(1-\gamma)n}(\vec{y}_Q|\vec{x}_Q(\vec{y}_P,0)), \nonumber \\ 	
            &\qquad\qquad P^{\gamma n}(\vec{y}_P|\vec{x}_P(1))Q^{(1-\gamma)n}(\vec{y}_Q|\vec{x}_Q(\vec{y}_P,1)) \big) \\
        &\ge \frac{1}{2} \min_{\vec{x}_P,\vec{x}'_P} \sum_{\vec{y}_P} \min_{\vec{x}_Q,\vec{x}'_Q} \sum_{\vec{y}_Q} \min\big( P^{\gamma n}(\vec{y}_P|\vec{x}_P)Q^{(1-\gamma)n}(\vec{y}_Q|\vec{x}_Q), \nonumber \\
            &\qquad\qquad P^{\gamma n}(\vec{y}_P|\vec{x}'_P)Q^{(1-\gamma) n}(\vec{y}_Q|\vec{x}'_Q)  \big). \label{eq:sum_order}
    \end{align}
    We now observe that this expression coincides with the non-feedback case given in \eqref{eq:no-fb}, except that the minimum over $(\vec{x}_Q,\vec{x}'_Q)$ comes after the summation over $\vec{y}_P$ (since the former is allowed to depend on the latter).  
    
    While this ordering can be significant in general, we now show that it has no effect when $Q$ is a BSC.  Without loss of optimality, we can assume that $\vec{x}_Q$ and $\vec{x}'_Q$ differ in every entry, since any entries that coincide would be independent of $\Theta$ (at least conditionally given $\vec{y}_P$) and reveal no information (or viewed differently, could be simulated at the decoder).  However, when $Q$ is a BSC, its symmetry implies that any quantity of the form 
    \begin{equation}
        \sum_{\vec{y}_Q} \min( a Q^{(1-\gamma)n}(\vec{y}_Q|\vec{x}_Q), b Q^{(1-\gamma)n}(\vec{y}_Q|\vec{x}'_Q) )
    \end{equation}
    (with constants $a$ and $b$) is identical for all such $(\vec{x}_Q,\vec{x}'_Q)$ pairs; the precise choice only amounts to re-ordering terms in the summation over $\vec{y}_Q$.  It follows that interchanging the order of $\sum_{\vec{y}_P}$ and $\min_{\vec{x}_Q,\vec{x}'_Q}$ in \eqref{eq:sum_order} has no impact, and we recover the non-feedback expression \eqref{eq:no-fb}, whose error exponent is defined to be $E_1(P,Q,\gamma)$.
    
    It only remains to show that the error exponent $E_1(P,Q,\gamma)$ is no higher than the learning rate in \eqref{dmc-learning-rate}.  To show this, we re-use some of the findings from Appendix \ref{app:non_fb_exponent}.\footnote{Assumption \ref{assump1} was not applied until {\em after} Lemma \ref{lem:gamma} in Appendix \ref{app:non_fb_exponent}, and we will not use it here.}  In particular, the error exponent is still given by \eqref{eq:lim_mu}, and $\mu_{n,\gamma}(s)$ still satisfies \eqref{86}, assuming without loss of optimality that the two codewords $\vec{x}$ and $\vec{x}'$ differ in every entry.
    
    For any fixed $s$, we can make $\mu_{n,\gamma}(s)$ in \eqref{86} smaller or equal by choosing both $\beta_P$ and $\beta_Q$ to be either zero or one (e.g., if $\mu_P(s) < \mu_P(1-s)$ then we set $\beta_P= 1$).  That is, without loss of optimality, every input to  $P$ is identical, and every input to $Q$ is identical.  Assuming without loss of generality that $\beta_P = 1$, and noting that the symmetry of the BSC implies $\mu_Q(s) = \mu_Q(1-s)$, it follows that we can simplify \eqref{86} to
    \begin{equation}
        \mu_{n,\gamma}(s) = \gamma n \mu_P(s) + (1-\gamma)n \mu_Q(s).
    \end{equation}
    This precisely coincides with the quantity $\mutilde_{n,\gamma}(s)$ introduced in \eqref{eq:mutilde}, and hence, we can directly apply Lemma \ref{lem:gamma} to obtain $\mu_{n,\gamma}(s) \ge n \mu^*$ for some $\gamma \in [0,1]$ and all $s \in [0,1]$, where $\mu^*$ is defined in \eqref{87}. Since $\mu^*$ coincides with \eqref{dmc-learning-rate}, this completes the proof.
    
    \subsection{Proof of Corollary \ref{cor:strict_improvement} (Strict Improvement in the Converse)} \label{app:strict_improvement}
    
    To see that the converse in Theorem \ref{thm:conv} can strictly improve on the trivial one in \eqref{naive-bound}, we re-use the example shown in Figure \ref{fig:fb_exponent} in Appendix \ref{app:fb_example}, but with $P$ and $Q$ interchanged (i.e., $P$ is a reverse Z-channel with parameter $0.8$, and $Q$ is a BSC with parameter $0.2$).  In this case, it is straightforward to verify that the ``No Feedback'' curve in Figure \ref{fig:fb_exponent} is simply flipped with respect to the mid-point $\gamma = \frac{1}{2}$, and in particular, all values of $\gamma \in (0,1)$ still give a strictly smaller exponent than the endpoints $\gamma = 0$ and $\gamma  = 1$.
    
    The key difference due to interchanging $P$ and $Q$ is that in view of the above proof of Lemma \ref{lem:symmQ}, the feedback no longer helps, and hence, the same curve just discussed serves as a valid converse even with the single round of feedback.  Hence, by Theorem \ref{thm:conv} and the fact the simple converse in \eqref{naive-bound} corresponds to taking the better of $\gamma=0$ and $\gamma=1$, we conclude that a strict improvement has indeed been shown.
    
    While the reliance of the preceding proof on a numerical evaluation can be circumvented, we omit such an approach, since we believe that this example has no ambiguity with respect to potential numerical issues.
    
    \subsection{Proof of Theorem \ref{thm:block_conv} (Converse for General Block-Structured Protocols)} \label{app:block_converse}
    
    We first consider a single block.  Letting $Y$ denote the $k$ symbols received by the teacher, $\hat{X}$ denote the $k$ transmitted bits, and $W$ denote the block received by the student (here we omit the vector notation $\vec{(\cdot)}$ to lighten the exposition), we have the Markov chain $\Theta \rightarrow Y \rightarrow \hat{X} \rightarrow W$.
    
    We begin with an inequality on mixture distributions ({\em cf.}, Definition \ref{def:mixture}). Let $A$ and $B$ be probability distributions, and suppose $A$ follows a mixture distribution described by $[p_1,A_1; \ldots; p_k,A_k]$. Then, observe that
    \begin{align}
        \rho(A,B,s) 
        &= \sum_{x} \bigg(\sum_{i=1}^m p_i \p(A_i = x)\bigg)^{1-s} \p(B=x)^s \label{eq:mixture1} \\
        &\ge \sum_{i=1}^m p_i \sum_{x} \p(A_i = x)^{1-s} \p(B=x)^s  \label{eq:mixture2}  \\
        &\ge \min_{i=1,\dotsc,m} \rho(A_i,B,s),  \label{eq:mixture3} 
    \end{align}
    where \eqref{eq:mixture2} applies Jensen's inequality to the concave function $f(z) = z^{1-s}$ ($s \in [0,1]$), and \eqref{eq:mixture3} lower bounds the average by the minimum.  The same argument applies when $A$ is fixed and $B$ is a mixture.
    
    For convenience, define $\mu(A,B,s) = \ln \rho(A,B,s)$.  Considering two length-$k$ sequences $\vec{x},\vec{x}'$, by the additive property of $\mu$ (i.e., multiplicative property of $\rho$) for product distributions, the corresponding conditional distributions for $W$ satisfy
    \begin{equation}
        \mu(P_{W|\hat{X}=\vec{x}}, P_{W|\hat{X}=\vec{x}'},s) \geq  k_1\mu_Q(s) + k_2 \mu_Q(1-s),
        \label{decomposed-bound}
    \end{equation}
    where $k_1$ is the number of indices where $\vec{x}$ is zero and $\vec{x}'$ is one, and vice versa for $k_2$.   
    
    Supposing first that $s \in [0,\frac{1}{2}]$, we can use the assumption $\mu_Q(s) \leq \mu_Q(1-s)$ from Assumption \ref{assump1}, and weaken \eqref{decomposed-bound} to 
    \begin{equation}
        \mu(P_{W|\hat{X}=\vec{x}}, P_{W|\hat{X}=\vec{x}'},s) \geq  k \mu_Q(s)
        \label{decomposed-bound2}
    \end{equation}
    since $k_1+k_2 \leq k$.  Observe that \eqref{decomposed-bound2} gives a lower bound that holds uniformly with respect to $\vec{x},\vec{x}'$. Then, using the property \eqref{eq:mixture3} and the fact that $P_{W|\Theta=0}$ and $P_{W|\Theta=1}$ are mixture distributions mixed by $\hat{X}$, we see that \eqref{decomposed-bound2} implies (for $s \in [0,\frac{1}{2}]$) that
    \begin{equation}
        \mu(P_{W|\Theta=0}, P_{W|\Theta=1},s) \geq k \mu_Q(s). \label{eq:Qdep}
    \end{equation}
    This gives us a $Q$-dependent lower bound, and the $P$-dependent lower bound turns out to be simpler: Using the Markov chain $\Theta \rightarrow Y \rightarrow \hat{X}$, along with the data processing inequality for $\mu(\cdot,\cdot,s)$,\footnote{The data processing inequality for $\mu$ is equivalent to that for R\'enyi divergence, defined as $D_{\alpha}(P\|Q) = \frac{1}{\alpha - 1}\ln\sum_{x} P(x)^{\alpha} Q(x)^{1-\alpha}$.  For the data processing inequality for R\'enyi divergence, see for example \cite[Example 2]{Van14}.} we have
    \begin{equation}
        \mu(P_{W|\Theta=0}, P_{W|\Theta=1},s) \geq \mu(P_{Y|\Theta=0}, P_{Y|\Theta=1},s) = k \mu_P(s). \label{eq:data_proc}
    \end{equation}
    Combining the lower bounds, we have for all $s \in [0,\frac{1}{2}]$ that
    \begin{equation}
        \mu(P_{W|\Theta=0}, P_{W|\Theta=1}, s) \geq k\max(\mu_P(s), \mu_Q(s)).
    \end{equation}
    
    If $s \in [\frac{1}{2},1]$, then Assumption \ref{assump1} gives $\mu(s)\geq \mu(1-s)$, and the analog of \eqref{eq:Qdep} is
    \begin{equation}
        \mu(P_{W|\Theta=0}, P_{W|\Theta=1},s) \geq \mu_Q(1-s).
    \end{equation}
    Since \eqref{eq:data_proc} holds for all $s \in [0,1]$, it follows that for any $s \in [\frac{1}{2},1]$, we have
    \begin{equation}
        \mu(P_{W|\Theta=0}, P_{W|\Theta=1},s) \geq k\max(\mu_P(1-s), \mu_Q(1-s)).
    \end{equation}
    Combining the cases $s \in [0,\frac{1}{2}]$ and $s \in [\frac{1}{2},1]$, we deduce that
    \begin{equation}
        \min_{s \in [0,1]} \mu(P_{W|\Theta=0}, P_{W|\Theta=1},s) \geq     \min_{s \in [0,\frac{1}{2}]} k\max(\mu_P(s), \mu_Q(s)). \label{eq:single_block}
    \end{equation}
    
    Having studied a single block, we are now in a position to study the overall combination of $\frac{n}{k}-1$ blocks,\footnote{The first block received by the student does not depend on $\Theta$, and provides no distinguishing power.} defining the $n$-letter version of $\mu$ as
    \begin{equation}
        \mu_n(s) = \ln \sum_{\vec{z}} \p( \vec{Z} = \vec{z} \,|\, \Theta=0 )^{1-s} \p( \vec{Z} = \vec{z} \,|\, \Theta=1 )^{s},
    \end{equation}  
    where $\vec{Z}$ is the entire length-$n$ sequence received by the student.  Assuming momentarily that the teacher employs an identical strategy within each block, the additive property of $\mu$ for product distributions (arising from $P$ and $Q$ being memoryless) gives
    \begin{equation}
        \mu_n(s) = \left(\frac{n}{k}-1\right) \mu(P_{W|\Theta=0}, P_{W|\Theta=1},s), \label{eq:mu_n_lb0}
    \end{equation}
    and substituting \eqref{eq:single_block} gives
    \begin{equation}
        \min_{s \in [0,1]} \mu_n(s) \geq \min_{s \in [0,\frac{1}{2}]} (n-k)\max(\mu_P(s), \mu_Q(s)). \label{eq:mu_n_lb}
    \end{equation}
    In fact, this lower bound holds even if the teacher applies a different strategy between blocks, since we showed that \eqref{eq:single_block} holds regardless of the strategy used within the block.
    
    Letting $s^* = \argmin_{s \in [0,1]} \mu_n(s)$, we have from \cite[Thm.~5]{shannon1967} (a more general form of Lemma \ref{lem:berlekamp}) that, for any decoding rule employed by the student, the error probability is lower bounded by
    \begin{equation}
        \mathbb{P}(\hat{\Theta}_n \neq \Theta) \geq \frac18\exp(\mu_n(s^*)-\sqrt{2\mu''_n(s^*)}),
        \label{error-after-n-steps}
    \end{equation}
    where $\mu''_n(s)$ denotes the second derivative when $s \in (0,1)$, and the appropriate limiting value when $s \in \{0,1\}$.  Again using the additive property of $\mu$, the following holds if the teacher employs the same per-block strategy in each block:
    \begin{equation}
        \mu_n''(s) = \Big(\frac{n}{k} - 1\Big) \mu''(P_{W|\Theta=0}, P_{W|\Theta=1},s). \label{eq:mu''}
    \end{equation}
    Once again, our analysis also extends immediately to the case of varying per-block strategies.
    
    In Lemma \ref{lem:quadratic} below, we show that $\mu''(P_{W|\Theta=0}, P_{W|\Theta=1},s) = O(k^2)$, and substitution into \eqref{eq:mu''} gives $\mu''_n(s) = O(nk)$.  The assumption $k = o(n)$ then yields $\mu''_n(s) = o(n^2)$, or $\sqrt{\mu''_n(s)} = o(n)$.  Substituting this scaling into \eqref{error-after-n-steps}, applying \eqref{eq:mu_n_lb}, and taking suitable limits, we obtain
    \begin{equation}
        \limsup_{n\rightarrow \infty}\, -\frac{1}{n} \ln \mathbb{P}(\Theta_n \neq \Theta) \leq - \min_{s \in [0,\frac{1}{2}]} \max(\mu_P(s), \mu_Q(s)). \label{eq:pf_final}
    \end{equation}
    We deduce Theorem \ref{thm:block_conv} by further upper bounding the right-hand side by expanding the minimum from $[0,\frac{1}{2}]$ to $[0,1]$; \eqref{eq:pf_final} additionally reveals that further restricting $s \in [0,\frac{1}{2}]$ is without loss of optimality, which is unsurprising given Assumption \ref{assump1}.
    
    It only remains to show the following.
    
    \begin{lemma} \label{lem:quadratic}
        For any fixed $P$ and $Q$, and any $s \in [0,1]$, it holds that $\mu''(P_{W|\Theta=0}, P_{W|\Theta=1},s) = O(k^2)$.
    \end{lemma}
    \begin{proof}
        To reduce notation, define $P_0 = P_{W|\Theta=0}$ and $P_1 = P_{W|\Theta=1}$.  It is shown in \cite[p.~85]{shannon1967} that $\mu''$ can be written as the variance of a log-likelihood ratio:
        \begin{equation}
            \mu''(P_0, P_1,s) = {\rm Var}\bigg[ \log \frac{P_1(\tilde{W})}{P_0(\tilde{W})} \bigg],
        \end{equation}
        where $\tilde{W}$ is distributed according to the following ``tilted'' distribution when $s \in (0,1)$:
        \begin{equation}
            \tilde{P}(\vec{w}) = \frac{ P_0(\vec{w})^{1-s} P_1(\vec{w})^s }{ \sum_{\vec{w}'} P_0(\vec{w}')^{1-s} P_1(\vec{w}')^s }. \label{eq:Ptilde}
        \end{equation}
        In addition, when $s$ equals an endpoint (i.e., 0 or 1), we can use the same expression for $\tilde{P}$, except that $\vec{w}$ must be restricted to satisfy both $P_0(\vec{w}) > 0$ and $P_1(\vec{w}) > 0$ (with $\tilde{P}(\vec{w}) = 0$ otherwise) \cite{shannon1967}.  
        
        We upper bound the variance by the second moment, i.e., $\e\big[ \big( \log \frac{P_1(\tilde{W})}{P_0(\tilde{W})} \big)^2 \big]$, and proceed by further upper bounding the latter.  By the definition of $\tilde{P}$, any $\vec{w}$ such that $P_0(\vec{w}) = 0$ or $P_1(\vec{w}) = 0$ does not contribute to the second moment.  On the other hand, for $\theta \in \{0,1\}$, if $P_{\theta}(\vec{w}) \ne 0$, then we have
        \begin{align}
            P_{\theta}(\vec{w}) 
            &= \sum_{\vec{y}} P^k(\vec{y}|\theta) Q^k(\vec{w}|\vec{\hat{x}}(\vec{y})) \\
            &\ge P_{\min}^k Q_{\min}^k, \label{eq:PminQmin}
        \end{align}
        where $P^k$ is the $k$-fold product of $P$ (and similarly for $Q^k$), $\vec{\hat{x}}(\vec{y})$ is the transmitted $\hat{X}$ sequence when the teacher receives $\vec{y}$, and $P_{\min},Q_{\min}$ are the smallest non-zero transition probabilities of $P$ and $Q$.  It follows that each non-zero $P_{\theta}(\vec{w})$ is bounded between $P_{\min}^k Q_{\min}^k$ and $1$, which implies that $\log \frac{P_1(\vec{w})}{P_0(\vec{w})} = O(k)$, and hence $\e\big[ \big( \log \frac{P_1(\tilde{W})}{P_0(\tilde{W})} \big)^2 \big] = O(k^2)$.
        
    \end{proof}
    
    \section*{Acknowledgment}
    
    This work was supported by the Singapore National Research Foundation (NRF) under grant number R-252-000-A74-281.
    
    \bibliographystyle{IEEEtran}
    \bibliography{general}
    
     \begin{IEEEbiographynophoto}{Yan Hao Ling}
        received the B.Comp.~degree in computer science and the B.Sci.~degree 
        in mathematics from the National University of Singapore (NUS) in 2021. 
        He is now a PhD student in the Department of Computer Science at NUS.
        His research interests are in the areas of
        information theory, statistical learning, and theoretical computer science.
    \end{IEEEbiographynophoto}
    
     \begin{IEEEbiographynophoto}{Jonathan Scarlett}
        (S'14 -- M'15) received 
        the B.Eng.~degree in electrical engineering and the B.Sci.~degree in 
        computer science from the University of Melbourne, Australia. 
        From October 2011 to August 2014, he
        was a Ph.D. student in the Signal Processing and Communications Group
        at the University of Cambridge, United Kingdom. From September 2014 to
        September 2017, he was post-doctoral researcher with the Laboratory for
        Information and Inference Systems at the \'Ecole Polytechnique F\'ed\'erale
        de Lausanne, Switzerland. Since January 2018, he has been an assistant
        professor in the Department of Computer Science and Department of Mathematics,
        National University of Singapore. His research interests are in
        the areas of information theory, machine learning, signal processing, and
        high-dimensional statistics. He received the Singapore National Research Foundation (NRF)
        fellowship, and the NUS Early Career Research Award.
    \end{IEEEbiographynophoto}
    
\end{document}